\documentclass[10pt]{amsart}
\usepackage[margin=1.5in]{geometry}
\pdfoutput=1

\usepackage{colortbl}
\usepackage{enumitem}

\newtheorem{theorem}{Theorem}[section]
\newtheorem{lemma}[theorem]{Lemma}

\newtheorem{corollary}[theorem]{Corollary}
\theoremstyle{definition}
\newtheorem{definition}[theorem]{Definition}

\theoremstyle{remark}
\newtheorem{remark}[theorem]{Remark}

\numberwithin{equation}{section}

\usepackage{graphicx}
\usepackage[export]{adjustbox}
\usepackage{caption}
\usepackage{subcaption}
\usepackage{hyperref}

\begin{document}

\title[A particle model for the herding phenomena]{A particle model for the herding phenomena induced by dynamic market signals}

\author{Hyeong-Ohk Bae}
\address{Department of Financial Engineering, Ajou University, Suwon, Republic of Korea}
\email{hobae@ajou.ac.kr}
\author{Seung-yeon Cho}
\address{Department of Mathematics, Sungkyunkwan University, Suwon 440-746, Republic of Korea}
\email{chosy89@skku.edu}
\author{Sang-hyeok Lee}
\address{Department of Financial Engineering, Ajou University, Suwon, Republic of Korea}
\email{sanghyeoke@ajou.ac.kr}
\author{SEOK-BAE YUN }
\address{Department of Mathematics, Sungkyunkwan University, Suwon 440-746, Republic of Korea}
\email{sbyun01@skku.edu}




\thanks{H.-O. Bae was supported by the Basic Science Research Program through the National Research Foundation of Korea (NRF) funded by the Ministry of Education, Science and Technology (2015R1D1A1A01057976). S.-B. Yun
was supported by the Basic Science Research Program through the National Research Foundation of Korea (NRF) funded by the Ministry of Education, Science and Technology by
(NRF-2016R1D1A1B03935955).} 
\begin{abstract}
In this paper, we study  the herding phenomena in financial markets arising from the combined effect of (1) non-coordinated collective interactions between the market players and
(2) concurrent reactions of market players to dynamic market signals.
By interpreting the expected rate of return of an asset and the favorability on that asset  as position and velocity in phase space, we construct  an agent-based particle model for herding behavior in finance.
We then define two types of herding functionals using this model, and show that they satisfy a Gronwall type estimate and a LaSalle type invariance property respectively, leading to the herding behavior of the market players. Various numerical tests are presented to numerically verify these results.
\end{abstract}
\maketitle

\section{introduction}

Collective behaviors  such as aggregation, fads, fashion, flocking and herding are frequently observed in nature \cite{BTTYB09, BBNS10,CCGPSST10,CS,TT} and society \cite{AZ98,BHW92,DW,HemSuk09}.  Among these various types of collective behaviors,  the flocking phenomena, in which alignment of the velocity occurs through the process of adjusting the velocity according to the particles around it, has seen tremendous progress recently.
 Several models have been suggested such as the Cucker-Smale model \cite{CS,CS2}
 or Viseck model \cite{VCBCS}, and many successful mathematical theories
  have been developed to understand those models \cite{CFRT,CHL,DM,HLL,HL,HT,MT}.


In this paper, we study the herding behavior arising in financial markets using a particle model. The term herding is used in several different contexts. The most basic meaning underlying them is the gathering behavior of individuals to form a group and move as a group. Therefore, unlike the flocking phenomena,
the adjustment occurs also between position variables and, therefore,  the
interactions between position and velocity may play important roles in the dynamics.
In finance and economy, herding is often used to describe the phenomena in which the market players tend more
and more to follow the market trend even though one's own opinion, information, favorability or instinct are against it \cite{AZ98}. The expression ``information cascade" is also frequently used.

In traditional economics and finance, it is assumed that all agents are rational, and
all the information is already reflected on the price (efficient markets hypothesis), which implies
the absence of bubbles \cite{ Abreu_B03, Milgrom_S82,Santos_W97,  Tirole82}.
However, as we can see in the examples of Tulip mania in 1637, the South Sea Bubble in 1711-1720,
the stock market boom in 2000, and the financial crisis in US housing market in 2007,
there have been many irrational events like bubbles and crashes.
It is not clear whether these phenomena are caused solely by herding behavior of the market players, but herding definitely plays a crucial role in the formulation of such phenomena.


Previous works on herding behavior in finance were largely based on sequential analysis
 \cite{Abreu_B03, AZ98, Banerjee92, BHW92, Brun01, HwangSalmon04, Lee95}, in which  the effect of the decision of the first player on the behavior of subsequent players was analyzed. (See Section 2.)
  In this paper, we study the herding phenomena arising from concurrent reaction to other players and dynamic market signals. By the word ``dynamic", we mean either (1) the signal changes with time, or (2) the signal is determined by the dynamics of market players.

%


For this, we introduce two variables: $x_i(t)$: the rate of return of assets over time $t$ expected by $i$th market player, and $v_i(t)$: the favorability that  $i$th market player has on those assets at time $t$. These two variables play the role of position and velocity of a self-propelled particle in the phase space, and enable one to derive a particle model giving dynamical relations between the rate of return, favorability and the market signal. (See Section {\ref{sec2}.)
We then analyze the herding properties of the system by deriving two Lyapunov type herding functionals satisfying a Gronwall type inequality and LaSalle type
invariance conditions respectively. (See Section 3, 4, 5.)

In flocking models,  the occupation of the same place by several particles is considered to be undesirable \cite{ACHL,CCMP,CS,CS2}.  In contrast,  we allow   ``particles" to occupy the same $x$ and $v$. Such overlap corresponds to the emergence of consensus on the expected  rate of return and the favorability on specific assets, which is exactly what we try to model.


The outline of this paper is as follows: In Section 2, we derive a particle model describing herding phenomena induced by dynamic market signals. A motivation of our model from the perspective of a pricing model in finance is also given.
In Section 3, our main herding theorems are presented. Section 4 and Section 5 are then devoted to the proof of the main results. In  Section 6, we provide some relevant numerical simulations.
The conclusion and possible future projects are discussed in Section 7.

\section{Particle model for herding behavior in finance}\label{sec2}
Suppose there are $N$ market players and $M$ assets, such as stocks or real estimates.
For simplicity we assume that no new players or assets enter or leave the market. We then define $x_i(t)=\big(x^1_i(t),\cdots, x^M_i(t) \big)\in \mathbb{R}^M$ and
$v_i(t)=\big(v_i^1(t),\cdots,v_i^M(t)\big)\in \mathbb{R}^M$,  $(i=1,\cdots,N)$  in the following way:
\begin{itemize}
\item $x_i(t)$:  the rate of return of the assets  $1,\cdots,M$  expected by  market player $i$ over time $t$.
\item $v_i(t)$:  the favorability that  market player $i$ has on those assets  $1,\cdots, M$ at time $t$.
\end{itemize}
We denote $x(t)=\big(x_1(t),x_2(t),\cdots,x_N(t)\big)\in\mathbb{R}^{MN}$ and $v(t)=\big(v_1(t),v_2(t),\cdots,v_N(t)\big)\in\mathbb{R}^{MN}$.\newline

It is natural to assume that the market players may consider the value of the asset in a more favorable way if the expectation on the rate of return rises,
 and the opposite when the rate of return decreases.
In this regard, we relate $x_i$ and $v_i$ by
\[
\frac{dx_i}{dt}=v_i. \quad (i=1,\cdots,N)
\]

To describe the dynamics of $v_i$, we assume that the market players are very sensitive to market trend and imitative strategy prevails in the market, which
is believed to be true by both economists and market participants up to certain level. We formulate this assumption by postulating that the favorability is affected by other player's assessment in the following three ways:
\begin{enumerate}
\item  Other players' evaluation on the expected rate of return for the assets:
\begin{align*}
\frac{1}{N}\sum_{j=1}^N\phi_{ij}(x_j-x_i).
\end{align*}
\item Other players' favorability on the assets:
\begin{align*}
 \frac{1}{N}\sum_{j=1}^N\phi_{ij}(v_j-v_i).
\end{align*}
In (1) and (2), $\phi_{ij}$ is the communication rate between player $i$ and player $j$, whose precise form will be given below.\newline
\item Various types of signals from the market also influence players' opinion and decision. 
    Such an influence can be observed more clearly when the market is experiencing a rapid transition or turmoil such as the 1997 Asian financial crisis and the subprime mortgage crisis in 2008, to name a few. To systematically formulate such signals, we introduce a function $w(x,t)$, which we call the ``dynamic market signal",  and assume that the favorability on the asset is affected by how big  the difference is between the expected return and the signal:
\begin{align*}
w(x,t)-x_i(t).
\end{align*}
Explicit examples of $w$ will be considered below.
\end{enumerate}

By combining the above three effects, we derive our main model:
\begin{align}\label{herding model}
\begin{split}
\frac{dx_i}{dt}&=v_i,\cr
\frac{dv_i}{dt}&
=\frac{\lambda_x}{N}\sum_{j=1}^N\phi_{ij}\big(x_j(t)-x_i(t)\big)
+\frac{\lambda_v}{N}\sum_{j=1}^N\phi_{ij}\big(v_j(t)-v_i(t)\big)
+\lambda_w \big(w(x,t)-x_i(t)\big),
\end{split}
\end{align}
where $\lambda_x$, $\lambda_v$ and $\lambda_w$ are interaction strength.

Several choices can be made for the communication rate $\phi_{ij}$, which determine how strongly a player's expected rate of return and  favorability
are influenced by those of other players in the market. Throughout this paper, we use
\[
\phi_{ij}:=\phi(x_i,x_j)=\frac{1}{\big(1+|x_i-x_j|^2\big)^\frac{\gamma}{2}}
\]
for $\gamma>0.$
We remark that the effect of noise should be considered for this model to be more realistic:
\begin{align*}
\begin{split}
\frac{dx_i}{dt}&=v_i,\cr
dv_i&=\frac{\lambda_x}{N}\sum_{j=1}^N\phi_{ij}\big(x_j(t)-x_i(t)\big)dt
+\frac{\lambda_v}{N}\sum_{j=1}^N\phi_{ij}\big(v_j(t)-v_i(t)\big)dt\cr
&+\lambda_w \Big(w(x,t)-x_i(t)\Big)dt+\sigma_i dW_t^i,
\end{split}
\end{align*}
where $\sigma_i$ is the volatility and $W_t^i$ is $M$-dimensional Brownian motion. Throughout this paper, however, we neglect the effect of noise and consider only $(\ref{herding model})$ for simplicity and clarity. We leave it as a future project.
%
%
%
%

\textbf{Examples of the dynamic market signal $w(x,t)$:}
We can choose various types of  mathematical expressions for the dynamic market signal depending on the market situation.
Some of the financially interesting examples are:
\begin{enumerate}[label=(\alph*)]
\item A sweeping trend of the market that the market players cannot handle. (Ex: abrupt upheaval in the market such as economic crisis or the rise and fall of foreign currency exchange rate) In this case, we set
    \[w=w(t)\]
    to be a given function of time.
\item  Asymmetric information or signal from an informed  influential market player like Warren Buffett.
Without loss of generality, we fix $x_1$ to be such influential  player so that we can set \[w(x,t)=x_1.
\]
That is, the rate of return expected by an influential player is a strong affecting factor in the market
\item  Average market expectation, market atmosphere or some average index  such as Dow Johns index, which is believed to  reflect  such  average market expectation, in which case we can define \[w(x,t)=\frac{1}{N}\sum^N_{i}x_i.\]
\end{enumerate}
We will show that the herding behavior induced by these signals can be explained in a unified manner. (See Section 3.)

\textbf{Financial motivation of the model}:
We now provide a financial motivation of our model (\ref{herding model}).
For that, we recall the geometric Brownian motion for an asset price $S(t)$:
\begin{align*}
\frac{dS}{S}=\mu dt + \sigma d W_t,
\end{align*}
where $\mu$ is the instantaneous expected rate of return,
$\sigma$ the volatility and $W_t$ the one-dimensional Brownian motion.
We then apply Ito's formula:
\begin{align*}
S(t)=S(0) e^{\int_{0}^{t}\mu(s)-\frac{1}{2}\sigma(s)^2 ds + \int_{0}^{t}\sigma(s) dW_s},
\end{align*}
and take expectation $\mathbb{E}[\cdot]$ to ${S(t)}$ to get
\begin{align*}
\mathbb{E}[S(t)]=S(0) e^{\int_{0}^{t}\mu(s)ds},
\end{align*}
or equivalently,
\begin{align}\label{logE and mu}
	\frac{d}{dt} \log{\mathbb{E}[S(t)]}  = \mu(t).
\end{align}

According to \cite{Merton}, the expected rate of return $\mu$ is generated by the following stochastic differential equation
\begin{align}\label{sde2}
	d\mu = a\big(\theta - \mu\big) dt + \delta\sigma dW_t,
\end{align}
where the first term represents a long-run regressive adjustment of the expected rate of return
toward a normal rate of return $\theta$ with the adjustment speed $a$, and
the second term is a short-run extrapolative adjustment of the expected rate of return of
the error-learning type with the adjustment speed $\delta$.
%
%
Coupling \eqref{sde2} with (\ref{logE and mu}),  we get the following system:
\begin{align*}
\frac{d}{dt} \log{\mathbb{E}[S]} = \mu,\qquad
d\mu = a(\theta - \mu)dt+ \delta\sigma dW_t.
\end{align*}
This can be generically extended to multi-variable system: $(i=1,\cdots, N)$
\begin{align*}
\frac{d}{dt} \log{\mathbb{E}[S_i]} = \mu_i,\qquad
d\mu_i = a(\theta_i - \mu_i)dt+ \delta\sigma_i dW_t^i.
\end{align*}
Rewriting $(\log{\mathbb{E}[S_i]}, \mu_i)$ by $(x_i, v_i)$ and taking $\delta=0$,
 we may express this as a particle model:
\begin{align*}
	x_i' = v_i,\qquad
	v_i' = a\big(\theta_i - v_i\big).
\end{align*}
Now, if we make the following choice:
\[
a=\lambda_v, ~\theta_i(t)=\frac{1}{N}\sum_{j=1}^N v_j,
\]
we recover our herding model (\ref{herding model}) with $\lambda_x=\lambda_w=0$ and $\phi_{ij}=1$.

\textbf{Brief review of the studies in Finance on Herding: } Herding behavior in finance has been extensively studied in the literature
including \cite{Abreu_B03, AZ98, Banerjee92, BHW92,
	Brun01,  Lee95}.
Herding behavior is associated with people blindly following the decisions of others \cite{Brun01}.
Imitating somebody's action can be rational if the predecessor's action affects one's
(1) payoff structure such that
imitation leads to a higher payoff (payoff externality) and/or
(2) his probability assessment of the state of the world such that it dominates the private signal
(informational externality).
Herding due to informational externalities occurs if an agent imitates the decision of his
predecessor even though his own signal might advise him to take a different action.
This herding can also lead to informational cascades.
In 
\cite{AZ98}, the concepts of investor herding and informational cascade are defined:
\begin{itemize}
\item An {\sl informational cascade} occurs  in a period $t$ when
$$P(h_t| V,H_t)=P(h_t| H_t)\qquad \forall V, h_t. $$
\item A trader with private information $x_\theta$ engages in {\sl herd behavior} at time $t$
 if he buys when $V^0_\theta (x_\theta)<V^0_m<V^t_m$
or if he sells when $V^0_\theta(x_\theta)>V^0_m>V^t_m$;
and buying (or selling) is strictly preferred to other actions.
\end{itemize}
Here,  $P(\cdot|\cdot)$ means the conditional probability,
$V$ denotes the value of the new information,
$H_t$ is the history of actions until time $t$,
and $h_t$ is the action (buy or sell) taken by the trader (market player) who arrives in period $t$,
$x_\theta$ is  the private information of trader $\theta$,
$V^t_m :=\mathbb {E}[V|H_t]$ is the market maker's expected value for the asset given public information,
which we sometimes refer to as the price,
and $V^t_\theta(x):=\mathbb{E}[V|H_t, x_\theta=x]$ is the expected value of an informed trader $\theta$.

In an informational cascade, new information on the asset does not affect the decision of the market players.
The above technical definition in \cite{AZ98} of the buying herding behavior can be expressed in the following three steps:
(1) Initially  a trader's evaluation is less than the market value of the asset, so that he is
inclined to sell.
(2) The market value of the asset is, nevertheless, increasing.
(3) The trader must want to buy the asset ignoring
his own evalution.
Also in \cite{AZ98}, it is shown that whether or not herd behavior affects asset prices,
asset prices can certainly affect herd behavior.

Even though it can make a large difference whether the market players decide sequentially or simultaneously,
most herding models are studied sequentially.
Up to the best knowledge of authors, particle model interpreting it as a dynamical system that reacts concurrently to other players and the market signal has not been proposed. A related study can be found in
\cite{ABHKL13,BHKLLY15}, where the flocking behavior of volatility is considered using Cucker-Smale type models.

As related works, we mention \cite{DL, DJT}, where Boltzmann type  kinetic equations were suggested to model the dynamics of a market and understand the formation of bubbles and crashes through the  combined effect of public information and  herding (See also \cite{Toscani}.), and \cite{BMP} in which a macroscopic herding model of Keller-Siegel type  was introduced to simulate  herding behaviours of human crowds.




%
%
%
%
%
%
\section{Main results}
In this section, we present our main results. We start with some simplification of our model for the convenience of  proof.
\subsection{Centralized herding model:} We first record a simple result on the averaged motion $x_i$
and $v_i$. Let $x_c$, $v_c$ denote the average expected rate of return and the average favorability
respectively:
\[
x_c(t)=\frac{1}{N} \sum_{i=1}^N x_i(t), \quad v_c(t)=\frac{1}{N} \sum_{i=1}^N v_i(t).
\]
These averaged quantities evolve according to the following simple system:
\begin{lemma}\label{average} $x_c$ and $v_c$ satisfy
\begin{align*}\label{mean}
\begin{split}
\frac{d}{dt}x_c(t)&=v_c(t),\\
\frac{d}{dt}v_c(t)&=\lambda_w\left( w(x,t)-x_c(t)\right).
\end{split}
\end{align*}
\end{lemma}
\begin{proof}
Since $\phi_{ij} = \phi_{ji}$ for all $1 \leq i,j \leq N$, we have from symmetry argument
\begin{equation}\label{symm}
\sum_{i=1}^N\sum_{j=1}^N\phi_{ij}(x_j-x_i)=\sum_{i=1}^N\sum_{j=1}^N\phi_{ij}(v_j-v_i)=0.
\end{equation}
Using these identity with $\frac{1}{N}\sum_{i=1}^{N} w(x,t)=w(x,t)$, we get the desired result by summing  (\ref{herding model}) over $1\leq i\leq N$.
\end{proof}
We will show that the large time behavior of $(x,v)$ is governed by $(x_c, v_c)$. In view of this, we replace
$$x_i(t)-x_c(t)\rightarrow \hat{x}_i(t), \quad v_i(t)-v_c(t) \rightarrow \hat{v}_i(t),$$
in our model (\ref{herding model}) to get
\begin{align}\label{centralized herding model}
\begin{split}
\frac{d\hat{x}_i}{dt}&=\hat{v}_i,\cr
\frac{d\hat{v}_i}{dt}&
=\frac{\lambda_x}{N}\sum_{j=1}^N\hat{\phi}_{ij}\big(\hat{x}_j(t)-\hat{x}_i(t)\big)
+\frac{\lambda_v}{N}\sum_{j=1}^N\hat{\phi}_{ij}\big(\hat{v}_j(t)-\hat{v}_i(t)\big)
-\lambda_w \hat{x}_i(t),
\end{split}
\end{align}
where $\hat{\phi}_{ij}:= \phi(\hat{x}_i,\hat{x}_j)$.
Note that
\begin{align}
\hat{x}_c(t)=\hat{v}_c(t)=0 \mbox{ for all } t \geq 0.
\end{align}
 From now on, we study only this centralized version for clarity and simplicity.

\subsection{Two herding functionals:} We define two kinds of herding energies $\mathcal{E}_1$ and $\mathcal{E}_2$ and prove their decay property. $\mathcal{E}_1$ is defined on a rather stringent assumptions on the parameters and initial configuration, but an explicit exponential
herding rate can be derived (Theorem \ref{main theorem 1}). For $\mathcal{E}_2$, such explicit herding rate is not available, but restrictions on the parameters and initial configuration can be relaxed a lot (Theorem \ref{main theorem 2}). To state our main theorem, we first need to define
some notations to be kept throughout this paper.
\begin{itemize}
\item $L^2$ deviation and covariance functionals:
\[
X(t)=\sum_{i=1}^{N}|\hat{x}_i(t)|^2,~ V(t)=\sum_{i=1}^{N}|\hat{v}_i(t)|^2
\]
and
\[
C(X,V)(t)=\sum_{i=1}^{N}\hat{x}_i(t)\cdot \hat{v}_i(t).
\]
\item Weighted $L^2$ deviation and weighted covariance functionals:
\[
X_{\phi}(t)=\frac{1}{N}\sum_{i=1}^{N}\sum_{j=1}^{N}\hat{\phi}_{ij}|\hat{x}_i(t)-\hat{x}_j(t)|^2,~ V_{\phi}(t)=\frac{1}{N}\sum_{i=1}^{N}\sum_{j=1}^{N}\hat{\phi}_{ij}|\hat{v}_i(t)-\hat{v}_j(t)|^2
\]
and
\[
C_{\phi}(X,V)(t)=\frac{1}{N}\sum_{i=1}^{N}\sum_{j=1}^{N}\hat{\phi}_{ij}(\hat{x}_i(t)-\hat{x}_j(t))(\hat{v}_i(t)-\hat{v}_j(t)).
\]
\end{itemize}
We define the herding behavior of market:
\begin{definition} Let $(\hat{x}(t),\hat{v}(t))$ be the solution to (\ref{centralized herding model}). Then, we say that the herding phenomena occurs if
\begin{align*}
	\lim\limits_{t\rightarrow \infty} |\hat{x}_i(t)-\hat{x}_j(t)|=0 \quad \text{and} \quad \lim\limits_{t\rightarrow \infty} |\hat{v}_i(t)-\hat{v}_j(t)|=0 \quad  \text{for all }~ i\ne j.
\end{align*}
\end{definition}
We now state our main theorem.
\subsection{Main result I -  Exponential herding:} Define the herding energy of the market by
\begin{align*}
\mathcal{E}_1(t)=\lambda_w X(t) + \frac{2\lambda^2_x}{(\lambda_x+\lambda_w)\lambda_v}C(X,V)(t) +V(t).
\end{align*}

\begin{theorem}\label{main theorem 1}
Let $\gamma>0$. Suppose the interaction strength $\lambda_x$, $\lambda_v$ and $\lambda_w$ satisfy
\begin{align}\label{parameter assumption}
\frac{1}{2}\left(\frac{\lambda_v}{\lambda_x}\right)^2\lambda_w>1.
\end{align}
Assume that the initial configuration  satisfies
\begin{align} \label{assumption}
\max_{i,j}|\hat{x}_i(0)-\hat{x}_j(0)|<\frac{1}{2}\sqrt{\left(\frac{1}{2}\left(\frac{\lambda_v}{\lambda_x}\right)^2\lambda_w\right)^{2/\gamma}-1}
\end{align}
and
\begin{align}\label{assumption2}
\mathcal{E}_1(0)<\frac{3}{2}\Big(1-\frac{M_0}{2}\Big)\lambda_w\max_{i,j}|\hat{x}_i(0)-\hat{x}_j(0)|^2,
\end{align}
where $M_0$	denotes
\begin{align}\label{M0}
M_0=\frac{1}{(1+4\max_{i,j}|\hat{x}_i(0)-\hat{x}_j(0)|^2)^{\frac{\gamma}{2}}}.
\end{align}
Then the herding functional $\mathcal{E}_1$ decays exponentially fast:
\[
\mathcal{E}_1(t)\leq e^{-\kappa t}\mathcal{E}_1(0),
\]
where the decay rate $\kappa$ is explicitly given by
\[
\kappa = \delta \min\bigg\{\frac{1}{2},-1+\frac{1}{2}\left(\frac{\lambda_v}{\lambda_x}\right)^2 \lambda_w M_0\bigg\}\frac{2\lambda^2_x}{(\lambda_x+\lambda_w)\lambda_v}
\]
and $0<\delta<1$ is a constant to be chosen in the proof. Moreover, we have
\begin{align}\label{2bae}
\max_{i,j}|\hat{x}_i(t)-\hat{x}_j(t)| < 2\max_{i,j}|\hat{x}_i(0)-\hat{x}_j(0)|
\end{align}
for all $t\geq 0$.
\end{theorem}
\begin{remark}\label{rmk1}
(1) $\mathcal{E}_1$ is always non-negative under the condition of Theorem \ref{main theorem 1}.
(2)  The set of $(x,v)$ satisfying the above conditions (\ref{assumption}) and (\ref{assumption2}) is non-empty. The proof of these remarks will be given at the end of Section 4.
\end{remark}
\noindent An immediate corollary is that the market shows an exponentially fast herding phenomena.
\begin{corollary}\label{cor3.5} Under the assumptions in Theorem \ref{main theorem 1}, herding occurs exponentially fast in the market:
\[
|\hat{x}_i(t)-\hat{x}_j(t)|^2\leq C_1e^{-\kappa t}\mathcal{E}_1(0),\quad |\hat{v}_i(t)-\hat{v}_j(t)|^2\leq C_2e^{-\kappa t}\mathcal{E}_1(0) \qquad (1\leq i,j \leq N)
\]
for
\[
C_1=\frac{4}{(2-M_0)\lambda_w}
\quad \text{and} \quad C_2=\frac{4M_0\lambda_w}{2M_0\lambda_w-\alpha^2}.
\]
Here, $\alpha$ denotes $$\alpha=\frac{2\lambda_x^2}{( \lambda_x+ \lambda_w)\lambda_v }$$
and $M_0$ is defined by (\ref{M0}).
\end{corollary}
\subsection{Main result II - Herding without decaying rate:}
We define another market energy $\mathcal{E}_2(t)$ that eventually vanishes even without any restrictions on  $\lambda_x,\lambda_v,\lambda_w>0$
and initial configurations. First, we define
\begin{align}
\displaystyle S_{\gamma}(t):=
\begin{cases}
\displaystyle \,\frac{1}{N}\sum_{i=1}^{N}\sum_{j=1}^N\bigg\{  \frac{1}{(1+|\hat{x}_i-\hat{x}_j|^2)^{\frac{r-2}{2}} }-1\bigg\} , & (0 < \gamma < 2)\\
\displaystyle \,\frac{1}{N}\sum_{i=1}^{N}\sum_{j=1}^N \log{(1+|\hat{x}_i-\hat{x}_j|^2)}, & (\gamma = 2)\\
\displaystyle \,\frac{1}{N}\sum_{i=1}^{N}\sum_{j=1}^N  \bigg\{1-\frac{1}{(1+|\hat{x}_i-\hat{x}_j|^2)^{\frac{r-2}{2}} }\bigg\}. &  (\gamma >2)\\
\end{cases}
\end{align}
Note that $S_{\gamma}(t) \geq 0$ for all $t\geq 0$.
Now, for any $\lambda_x, \lambda_v, \lambda_w>0$, we define the herding energy $\mathcal{E}_2(t)$ of the market by
\[
\mathcal{E}_2(t):=  \lambda_w X(t) +V(t) + \lambda_x\beta_{\gamma}^{-1}S_\gamma(t),
\]	
where $\beta_{\gamma} > 0$ is
\begin{align*}
\beta_\gamma=
\begin{cases}		
2-\gamma, & (0 <\gamma < 2)\\
2, & (\gamma = 2)\\
\gamma-2. & (\gamma > 2)
\end{cases}	
\end{align*}
\begin{theorem}\label{main theorem 2}
For any positive constants $\lambda_x$, $\lambda_v$ and $\lambda_w$ with initial data $\{x_i(0),v_i(0)\}_{i=1}^N$, $\mathcal{E}_2(t)$ goes to $0$ as $t$ tends to $\infty$.
\end{theorem}

\begin{remark}\label{rmk2}
(1) For any positive constants $\lambda_x$, $\lambda_v$ and $\lambda_w$, $\mathcal{E}_2$ is non-negative.
(2) Theorem \ref{main theorem 2} generalizes Theorem \ref{main theorem 1} in the sense that we do not impose any restriction on the initial configuration, nor on the parameters  $\lambda_x$, $\lambda_v$ and $\lambda_w$ except for positivity. Explicit herding rate, however, is not available in this case.
\end{remark}
\noindent This leads to the following general herding phenomena, which holds unconditionally.
\begin{corollary} Under the assumptions in Theorem \ref{main theorem 2}, the herding phenomena occurs in the market:
	\[
	\lim\limits_{t\rightarrow \infty} |\hat{x}_i(t)-\hat{x}_j(t)|=0,\quad\lim\limits_{t\rightarrow \infty} |\hat{v}_i(t)-\hat{v}_j(t)|=0.\qquad (1\leq i,j\leq N)
	\]
\end{corollary}	

%
%
%
%
%
%

\section{Proof of Theorem \ref{main theorem 1}: Exponentially fast herding }
Before we delve into the proof of the main theorem, we establish several technical lemmas.
\begin{lemma}\label{X,V}
We have
\begin{align*}
\frac{d}{dt}X(t)&=2C(X,V)(t),\cr
\frac{d}{dt}V(t)
&=-2\lambda_w C(X,V)(t)-\lambda_x C_{\phi}(X,V)(t)-\lambda_vV_{\phi}(t).
\end{align*}
\end{lemma}
\begin{proof}
The first identity is immediate:
\begin{align*}
\frac{d}{dt}\sum_{i=1}^{N}|\hat{x}_i(t)|^2	
&=2\sum_{i=1}^{N}\hat{x}_i(t)\cdot \hat{v}_i(t).
\end{align*}
For the second one, we compute
\begin{align*}
\frac{d}{dt}\sum_{i=1}^{N}|\hat{v}_i(t)|^2
&=2\sum_{i=1}^{N}\hat{v}_i(t) \cdot \bigg\{-\lambda_w \hat{x}_i(t)
+\frac{\lambda_x}{N}\sum_{j=1}^N\hat{\phi}_{ij}(\hat{x}_j(t)-\hat{x}_i(t))+\frac{\lambda_v}{N}\sum_{j=1}^N\hat{\phi}_{ij}(\hat{v}_j(t)-\hat{v}_i(t)) \bigg\}\\
&\qquad\equiv I+II+III.
\end{align*}
Then, clearly,	
\begin{align*}
I&=-2\lambda_w\sum_{i=1}^{N} \hat{x}_i(t) \cdot \hat{v}_i(t) .
\end{align*}
In view of (\ref{symm}), we have
\begin{align*}
II&=\frac{2\lambda_x}{N}\sum_{i=1}^{N} \sum_{j=1}^N\hat{\phi}_{ij} \hat{v}_i \cdot (\hat{x}_j-\hat{x}_i)\\
&=-\frac{\lambda_x}{N}\sum_{i=1}^{N}\sum_{j=1}^N\hat{\phi}_{ij}(\hat{x}_i-\hat{x}_j)\cdot (\hat{v}_i-\hat{v}_j).
\end{align*}
Similarly,
\begin{align*}
III&=\frac{2\lambda_v}{N}\sum_{i=1}^{N}\sum_{j=1}^N\hat{\phi}_{ij}\hat{v}_i \cdot (\hat{v}_j-\hat{v}_i)\\
&=-\frac{\lambda_v}{N}\sum_{i=1}^{N}\sum_{j=1}^N\hat{\phi}_{ij}|\hat{v}_i-\hat{v}_j|^2.\cr
\end{align*}
Therefore,
\begin{align*}
\frac{d}{dt}\sum_{i=1}^{N}|\hat{v}_i(t)|^2
&=-2\lambda_w\sum_{i=1}^{N}\hat{x}_i \cdot \hat{v}_i\\
&-\frac{\lambda_x}{N}\sum_{i=1}^{N}\sum_{j=1}^N\hat{\phi}_{ij}(\hat{x}_i-\hat{x}_j)\cdot (\hat{v}_i-\hat{v}_j)\\
&-\frac{\lambda_v}{N}\sum_{i=1}^{N}\sum_{j=1}^N\hat{\phi}_{ij}|\hat{v}_i-\hat{v}_j|^2.
\end{align*}
\end{proof}

We also need to consider the time evolution of $C(X,V)(t)$:
\begin{lemma}\label{C(X,V)}
For any $t>0$
\begin{align*}
\frac{d}{dt}C(X,V)(t)&= V(t) -\lambda_w X(t)
-\frac{\lambda_x}{2}X_{\phi}(t)
-\frac{\lambda_v}{2} C_{\phi}(X,V)(t).
\end{align*}
\end{lemma}
\begin{proof} See
\begin{align*}
\frac{d}{dt}\sum_{i=1}^{N}\hat{x}_i(t)\cdot \hat{v}_i(t)
&=\sum_{i=1}^{N}\frac{d\hat{x}_i}{dt} \cdot  \hat{v}_i+\hat{x}_i\cdot \frac{d\hat{v}_i}{dt} \cr
&\equiv I+II.
\end{align*}
Computation for $I$ is direct:	
\begin{align*}
I&=\sum_{i=1}^{N}|\hat{v}_i|^2.
\end{align*}
For $II$, we consider
\begin{align*}
II&=\sum_{i=1}^{N}\hat{x}_i \cdot \Big\{
-\lambda_w \hat{x}_i(t)+\frac{\lambda_x}{N}\sum_{j=1}^N\hat{\phi}_{ij}(\hat{x}_j(t)-\hat{x}_i(t))+\frac{\lambda_v}{N}\sum_{j=1}^N\hat{\phi}_{ij}(\hat{v}_j(t)-\hat{v}_i(t))\Big\}\cr
&\equiv II_1+II_2+II_3.
\end{align*}
We then compute each term. $II_1$ clearly is
\begin{align*}
II_1=-\lambda_w \sum_{i=1}^{N}|\hat{x}_i(t)|^2.
\end{align*}
From (\ref{symm}) and a simple symmetry argument, we get
\begin{align*}
II_2
&\equiv
\frac{\lambda_x}{N}\sum_{i}\sum_j\hat{\phi}_{ij}\hat{x}_i \cdot (\hat{x}_j-\hat{x}_i)\cr
&=
-\frac{\lambda_x}{N}\sum_{i}\sum_j\hat{\phi}_{ij}\hat{x}_i\cdot(\hat{x}_i-\hat{x}_j)\cr
&=
-\frac{\lambda_x}{2N}\sum_{i}\sum_j\hat{\phi}_{ij}|\hat{x}_i-\hat{x}_j|^2
\end{align*}
and
\begin{align*}
II_3
&\equiv
\frac{\lambda_v}{N}\sum_{i}\sum_j\hat{\phi}_{ij}\hat{x}_i \cdot (\hat{v}_j-\hat{v}_i)\cr
&=
-\frac{\lambda_v}{N}\sum_{i}\sum_j\hat{\phi}_{ij}\hat{x}_i \cdot (\hat{v}_i-\hat{v}_j)\cr
&=
-\frac{\lambda_v}{2N}\sum_{i}\sum_j\hat{\phi}_{ij}(\hat{x}_i-\hat{x}_j) \cdot (\hat{v}_i-\hat{v}_j),
\end{align*}
so that
\begin{align*}
II&=-\lambda_w \sum_{i=1}^{N}|\hat{x}_i|^2 \\
&-\frac{\lambda_x}{2N}\sum_{i}\sum_j\hat{\phi}_{ij}|\hat{x}_i-\hat{x}_j|^2\\
&-\frac{\lambda_v}{2N}\sum_{i}\sum_j\hat{\phi}_{ij}(\hat{x}_i-\hat{x}_j)(\hat{v}_i-\hat{v}_j).
\end{align*}
\end{proof}

%
%
%
%
%
\subsection{Proof of Theorem \ref{main theorem 1}:} We divide the proof into the following two steps:\newline
\noindent{\bf \underline{Step 1}:} Step 1 is devoted to the proof of the following claim:\newline

\noindent {\bf Claim:} Define $T>0$ by
\[
T=\sup\Big\{t\in\mathbb{R}_+\,\big|\,\max_{i,j}|\hat{x}_i(t)-\hat{x}_j(t)|\leq 2\max_{i,j}|\hat{x}_i(0)-\hat{x}_j(0)|\Big\}.
\]
Then we have
\begin{align*}
\frac{d}{dt}\mathcal{E}_1(t)
&\leq - \kappa \mathcal{E}_1(t)
\end{align*}
for $0\leq t<T$. Here $\kappa>0$ is a constant defined in the statement of Theorem \ref{main theorem 1}.

\begin{proof}
Recall the definition of $\alpha$ in Corollary \ref{cor3.5}:
$$\alpha:=\frac{2\lambda_x^2}{( \lambda_x+ \lambda_w)\lambda_v }.$$
Then, by Lemma \ref{X,V} and Lemma \ref{C(X,V)}, we have
\begin{align}\label{we have}
\begin{split}
\frac{d}{dt}\mathcal{E}_1(t) &= 2 \lambda_w C(X,V) +\alpha\bigg\{V -\lambda_w X- \frac{\lambda_x}{2}X_{\phi}- \frac{\lambda_v}{2} C_{\phi}(X,V)\bigg\} \\
&+ \bigg\{-2\lambda_w C(X,V)
- \lambda_x C_{\phi}(X,V) - \lambda_vV_{\phi}\bigg\}    \cr
&= \alpha V-\alpha\lambda_w X
-\frac{\alpha\lambda_x}{2}\bigg\{  X_{\phi} + \Big(\frac{\lambda_v}{\lambda_x} + \frac{2}{\alpha}\Big) C_{\phi}(X,V) + \frac{\lambda_v}{\lambda_x}\frac{2}{\alpha}V_{\phi}\bigg\}.
\end{split}
\end{align}
Now, using
\begin{align*}
	X(t)&=\sum_{i=1}^{N}\hat{x}_i\cdot \hat{x}_i\\
	&=\sum_{i=1}^{N}(\hat{x}_i-\hat{x}_c)\cdot \hat{x}_i \qquad (\hat{x}_c=0)\\
	&=\frac{1}{N}\sum_{i=1}^{N}\sum_{j=1}^{N}(\hat{x}_i-\hat{x}_j)\cdot \hat{x}_i\\
	&=\frac{1}{2N}\sum_{i=1}^{N}\sum_{j=1}^{N}|\hat{x}_i-\hat{x}_j|^2\\
&=\frac{1}{2N}\sum_{i=1}^{N}\sum_{j=1}^{N}(1-\hat{\phi}_{ij})|\hat{x}_i-\hat{x}_j|^2
+\frac{1}{2}X_{\phi}(t),	
\end{align*}
we obtain
\begin{align*}
-\alpha\lambda_w X&=-\frac{1}{2}\alpha\lambda_w X-\frac{1}{2}\alpha\lambda_w X\cr	
&=-\frac{1}{2}\alpha \lambda_w X -\frac{ \alpha\lambda_w}{4N}\sum_{i=1}^{N}\sum_{j=1}^{N}(1-\hat{\phi}_{ij})|\hat{x}_i-\hat{x}_j|^2-\frac{1}{4}\alpha\lambda_wX_{\phi}(t).	
\end{align*}
Therefore, we have from (\ref{we have})
\begin{align*}
\frac{d}{dt}\mathcal{E}_1(t)
&= \alpha V-\frac{1}{2}\alpha \lambda_w X -\frac{ \alpha\lambda_w}{4N}\sum_{i=1}^{N}\sum_{j=1}^{N}(1-\hat{\phi}_{ij})|\hat{x}_i-\hat{x}_j|^2 \\
&-\frac{\alpha \lambda_x}{2}\bigg\{  (1+\frac{ \lambda_w}{2\lambda_x})X_{\phi} + \Big(\frac{\lambda_v}{\lambda_x} + \frac{2}{\alpha}\Big) C_{\phi}(X,V) + \frac{\lambda_v}{\lambda_x}\frac{2}{\alpha}V_{\phi}\bigg\}  \cr
&\leq \alpha V-\frac{1}{2}\alpha \lambda_w X\\
&-\frac{\alpha \lambda_x}{2}\bigg\{  (1+\frac{ \lambda_w}{2\lambda_x})X_{\phi} + \Big(\frac{\lambda_v}{\lambda_x} + \frac{2}{\alpha}\Big) C_{\phi}(X,V) + \frac{\lambda_v}{\lambda_x}\frac{2}{\alpha}V_{\phi}\bigg\}  \cr
&\equiv\alpha V-\frac{1}{2}\alpha \lambda_w X+R.
\end{align*}
We set $L=1+\frac{\lambda_w}{2\lambda_x}$ for simplicity, and compute
\begin{align}\label{conse}
\begin{split}
R&=
-\frac{\alpha\lambda_x}{2}\bigg\{  LX_{\phi} + \Big(\frac{\lambda_v}{\lambda_x}+\frac{2}{\alpha}\Big) C_{\phi}(X,V) + \frac{2\lambda_v}{\lambda_x \alpha}V_{\phi}\bigg\}\\
&=
-\frac{\alpha\lambda_x}{2}\bigg\{  L\bigg(X_{\phi} + \frac{1}{L}\Big(\frac{\lambda_v}{\lambda_x}+\frac{2}{\alpha}\Big) C_{\phi}(X,V) + \frac{1}{4L^2}\Big(\frac{\lambda_v}{\lambda_x}+\frac{2}{\alpha})^2V_{\phi}\bigg)
+\frac{\lambda_v^2\lambda_w}{2\lambda_x^3} V_{\phi}\bigg\}\\
%
&=
-\frac{\alpha\lambda_x}{2}  \bigg\{\frac{L}{N}\sum_{i=1}^{N}\sum_{j=1}^{N}\hat{\phi}_{ij}\bigg((\hat{x}_i-\hat{x}_j) + \frac{1}{2L}\Big(\frac{\lambda_v}{\lambda_x}+\frac{2}{\alpha}\Big) (\hat{v}_i-\hat{v}_j) \bigg)^2+ \frac{\lambda_v^2\lambda_w}{2\lambda_x^3} V_{\phi}\bigg\}\cr
&\leq- \frac{\alpha\lambda_v^2\lambda_w}{4\lambda_x^2}  V_{\phi}.
\end{split}
\end{align}
Since we are assuming
\[
\max_{i,j}|\hat{x}_i(t)-\hat{x}_j(t)|\leq 2\max_{i,j}|\hat{x}_i(0)-\hat{x}_j(0)|,
\]
we have
\[
\hat{\phi}_{ij}\geq M_0
\]
for $M_0$ defined in (\ref{M0}).
Therefore,
\begin{align*}
V_{\phi}&=\frac{1}{N}\sum_{i=1}^{N}\sum_{j=1}^{N}\hat{\phi}_{ij}|\hat{v}_i-\hat{v}_j|^2\\
&\geq \frac{M_0}{N}\sum_{i=1}^{N}\sum_{j=1}^{N}|\hat{v}_i-\hat{v}_j|^2\\
&=\frac{2M_0}{N} \sum_{i=1}^{N}\sum_{j=1}^{N}(\hat{v}_i-\hat{v}_j)\hat{v}_i\\\
&=2M_0 \sum_{i=1}^{N}(\hat{v}_i-\hat{v}_c)\hat{v}_i\\
&=2M_0  \sum_{i=1}^{N}|\hat{v}_i|^2, \qquad (\hat{v}_c=0)
\end{align*}
which gives from (\ref{conse})
\begin{align*}
R\leq- \frac{\alpha\lambda_v^2\lambda_w}{2\lambda_x^2}  M_0 V.
\end{align*}
Finally, we go back to (\ref{we have}) with these computations to derive
\begin{align}\label{5151}
\begin{split}
\frac{d}{dt}\mathcal{E}_1(t)
&\leq -\frac{\alpha\lambda_w}{2} X -\bigg\{-1+\frac{\lambda_v^2\lambda_w}{2\lambda_x^2}  M_0\bigg\} \alpha V \\
&\leq -\frac{2\lambda^2_x}{(\lambda_x+\lambda_w)\lambda_v}
\min\bigg\{\frac{1}{2},-1+\frac{1}{2}\left(\frac{\lambda_v}{\lambda_x}\right)^2 \lambda_w M_0\bigg\} (\lambda_w X+V).
\end{split}
\end{align}
Next, we need to show that  there exists $\delta>0$ such that \\
\begin{align}\label{such that}
-(\lambda_w X+V)\leq -\delta(\lambda_w X+\alpha C(X,V)+V),
\end{align}
which is equivalent to find $\delta$ such that
\begin{align*}
(1-\delta)\left\{\lambda_w X-\frac{\delta}{1-\delta}\alpha C(X,V)+V\right\}&\geq 0.
\end{align*}
This holds if
\begin{align*}
\frac{\delta^2}{(1-\delta)^2}\alpha^2-4\lambda_w\leq 0.
\end{align*}
Recalling the definition of $\alpha$, it can be rewritten as
\begin{align*}
\frac{\delta^2}{(1-\delta)^2}  &\leq \lambda_w \frac{\lambda_v^2}{\lambda_x^2}\left(\frac{\lambda_w}{\lambda_x}+1\right)^2,
\end{align*}
which holds true for sufficiently small $\delta>0$. Now, with this choice of $\delta$, we can combine
(\ref{5151}) and (\ref{such that}) to close the desired Gronwall inequality:
\begin{align*}
\frac{d}{dt}\mathcal{E}_1(t)
&\leq -\frac{2\lambda^2_x}{(\lambda_x+\lambda_w)\lambda_v}
\min\bigg\{\frac{1}{2},-1+\frac{1}{2}\left(\frac{\lambda_v}{\lambda_x}\right)^2 \lambda_w M_0\bigg\} \delta \mathcal{E}_1(t).
\end{align*}	
\end{proof}

\label{pageTinf}
\noindent \noindent{\bf \underline{Step 2}:} We now prove that $T$ prolongs to infinity:\newline

\noindent {\bf Claim:} $T=\infty$, that is,
\[
\max_{i,j}|\hat{x}_i(t)-\hat{x}_j(t)|< 2\max_{i,j}|\hat{x}_i(0)-\hat{x}_j(0)|
\]
for all $t>0$.
\begin{proof}
We first rewrite (\ref{assumption}) as
\[
\frac{1}{2}\left(\frac{\lambda_v}{\lambda_x}\right)^2\lambda_w M_0>1
\]
and use the positivity of $\lambda_x$ and $\lambda_w$ to get
\[
2\lambda_w M_0>\frac{4\lambda_x^4}{(\lambda_x+ \lambda_w)^2\lambda_v^2}=\alpha^2,
\]
or, equivalently
\[
0<\frac{\alpha^2}{2M_0\lambda_w}<1.
\]
With this and the fact that  $0<M_0<2$, which follows from the definition of $M_0$, we see that
\begin{align}\label{positive}
\begin{split}
0 &\leq \big\{ 1-\frac{M_0}{2}\big\} \lambda_w X(t)\cr
&\leq\big\{1-\frac{M_0}{2}\big\}\lambda_w X
+ \frac{M_0 \lambda_w}{2}\sum_{i=1}^{N} \Big(\hat{x}_i + \frac{\alpha}{M_0\lambda_w}\hat{v}_i\Big)^2
+ \bigg\{1- \frac{\alpha^2}{2M_0\lambda_w}\bigg\}V \cr
&= \big\{1-\frac{M_0}{2}\big\}\lambda_w X + \frac{M_0\lambda_w}{2}\bigg\{X + \frac{2\alpha}{M_0\lambda_w}C(X,V) + \frac{\alpha^2}{M_0^2\lambda_w^2}V\bigg\} + \bigg\{1- \frac{\alpha^2}{2M_0\lambda_w}\bigg\}V \cr
&=\lambda_w X(t) + \frac{2\lambda^2_x}{(\lambda_x+\lambda_w)\lambda_v}C(X,V)(t) +V(t)\cr
&= \mathcal{E}_1(t).
\end{split}
\end{align}
Therefore, combining this with the result of Claim 1, we get,
\begin{align}\label{now}
X(t)\leq \frac{\mathcal{E}_1(0)}{(1-\frac{M_0}{2})\lambda_w }e^{-\kappa t}.
\end{align}
Now, contrary to the claim, suppose  that there are $i_0$ and $j_0$ such that
\begin{align}\label{contra}
|\hat{x}_{i_0}(T_0)-\hat{x}_{j_0}(T_0)|=2\max_{i,j}|\hat{x}_i(0)-\hat{x}_j(0)|
\end{align}
for some $T_{0}>0$. Then, for any $0\leq t\leq T_0$, we have from (\ref{now})
\begin{align*}
|\hat{x}_{i_0}(t)-\hat{x}_{j_0}(t)|^2 &\leq   2|\hat{x}_{i_0}(t)|^2 + 2|\hat{x}_{j_0}(t)|^2\cr
&\leq 2X(t)\cr
&\leq \frac{\mathcal{E}_1(0)}{(1-\frac{M_0}{2})\lambda_w }e^{-\kappa t}.
\end{align*}
Therefore, applying the result of Step I, we deduce for $0\leq t\leq T_0$
\begin{align*}
|\hat{x}_{i_0}(t)-\hat{x}_{j_0}(t)|^2 &\leq  \frac{2\mathcal{E}_1(0)}{(1-\frac{M_0}{2})\lambda_w }< 3\max_{i,j}|\hat{x}_i(0)-\hat{x}_j(0)|^2.
\end{align*}
The last inequality is from (\ref{assumption2}):	
\[
\mathcal{E}_1(0)<\frac{3}{2}\left(1-\frac{M_0}{2}\right)\lambda_w\max_{i,j}|\hat{x}_i(0)-\hat{x}_j(0)|^2.
\]
In conclusion, we have
\begin{align*}
|\hat{x}_{i_0}(t)-\hat{x}_{j_0}(t)|&<\sqrt{3}\max_{i,j}|\hat{x}_i(0)-\hat{x}_j(0)|
\end{align*}
for all $0\leq t\leq T_0$, which  is contradictory to (\ref{contra}). Therefore, $T=\infty$. This completes the proof.	
\end{proof}		

\subsection{Proof of Corollary \ref{cor3.5}:}
By definition of $X$, we have
\[
|\hat{x}_i(t)-\hat{x}_j(t)|^2 \leq 2X(t).
\]
Then (\ref{now}) gives the desired result. The proof for $V$ is similar.\newline
\subsection{Proof of Remark \ref{rmk1}}
\noindent(1) It was shown in (\ref{positive}) that $\mathcal{E}_1(t)$ is positive under our assumptions in Theorem \ref{main theorem 1}.\newline
\noindent(2) We only consider the case where there are at least two players in the market, that is $N \geq 2$. For this, we define $\eta$ by
	\begin{align*}
		\eta \equiv \frac{1}{2}\sqrt{\left(\frac{1}{2}\left(\frac{\lambda_v}{\lambda_x}\right)^2\lambda_w\right)^{2/\gamma}-1}>0,
	\end{align*}
	where the positivity of $\eta$ comes from (\ref{parameter assumption}). With this $\eta$, we choose $x_1(0)$ and $x_2(0)$ as
	\begin{align*}
		x_1(0)= \frac{1}{3} \eta \quad \text{and}	\quad x_2(0)= -\frac{1}{3} \eta,
	\end{align*}
	and set all the remaining $x_i(0)$ and $v_j(0)$ to be zero. Since $x_c(0)=0$ and $v_c(0)=0$, we can say $x_i(0)=\hat{x}_i(0)$ and $v_i(0)=\hat{v}_i(0)$. Then, we have
	\begin{align*}
		\max_{i,j}|\hat{x}_i(0)-\hat{x}_j(0)| = \frac{2}{3} \eta <\eta.
	\end{align*}
	Moreover, $\hat{v}(0)=0$ implies $C(X,V)(0)=V(0)=0$ and $\hat{x}_1(0)+\hat{x}_2(0)=0 $ does $|\hat{x}_1(0)| = |\hat{x}_2(0)|$. Therefore,
	\begin{align*}
		\mathcal{E}_1(0)&= \lambda_w \big(|\hat{x}_1(0)|^2 + |\hat{x}_2(0)|^2 \big)  \cr
		&= 2\lambda_w |\hat{x}_1(0)|^2  \cr
		&= \frac{2}{4}\lambda_w |2\hat{x}_1(0)|^2  \cr
		&<\frac{3}{2}\Big(1-\frac{M_0}{2}\Big)\lambda_w\max_{i,j}|\hat{x}_i(0)-\hat{x}_j(0)|^2.
	\end{align*}
	In the last line, we used
	\begin{align*}
		0 < M_0 \leq 1
	\end{align*}
	and
	\begin{align*}
		|2\hat{x}_1(0)|^2 &=\max_{i,j}|\hat{x}_i(0)-\hat{x}_j(0)|^2.
	\end{align*}

%
%
%
%
%
%
\section{Proof of Theorem \ref{main theorem 2}: Herding without explicit decay rate}
We start with establishing technical lemmas.
\begin{lemma}
For any $\gamma > 0$, we have
\begin{align}\label{new energy}
\frac{dS_{\gamma}(t)}{dt}= \beta_\gamma C_{\phi}(X,V)(t),
\end{align}
where $\beta_{\gamma} > 0$ is given by
\begin{align*}
\beta_\gamma=
\begin{cases}		
2-\gamma, & (0 <\gamma < 2)\\
2, & (\gamma = 2)\\
\gamma-2. & (\gamma > 2)
\end{cases}
\end{align*}

\end{lemma}
\begin{proof}
When $ \gamma \neq 2$, we have
\begin{align*}
\frac{d}{dt}\sum_{i=1}^{N}\sum_{j=1}^N \frac{1}{(1+|\hat{x}_i-\hat{x}_j|^2)^{\frac{r-2}{2}} }
&=(2-\gamma)\sum_{i=1}^{N}\sum_{j=1}^N\frac{(\hat{x}_i-\hat{x}_j)(\hat{v}_i-\hat{v}_j)}{\big(1+|\hat{x}_i-\hat{x}_j|^2\big)^\frac{\gamma}{2}}\cr
&=(2-\gamma)\sum_{i=1}^{N}\sum_{j=1}^N\hat{\phi}_{ij}(\hat{x}_i-\hat{x}_j)(\hat{v}_i-\hat{v}_j).
\end{align*}
For $ \gamma = 2$, we similarly compute
\begin{align*}
\frac{d}{dt}\sum_{i=1}^{N}\sum_{j=1}^N \log{(1+|\hat{x}_i-\hat{x}_j|^2)}
&=2\sum_{i=1}^{N}\sum_{j=1}^N\frac{(\hat{x}_i-\hat{x}_j)(\hat{v}_i-\hat{v}_j)}{(1+|\hat{x}_i-\hat{x}_j|^2)}\cr
&=2\sum_{i=1}^{N}\sum_{j=1}^N\hat{\phi}_{ij}(\hat{x}_i-\hat{x}_j)(\hat{v}_i-\hat{v}_j).
\end{align*}
\end{proof}

\begin{lemma}\label{decay} \label{E_2'''} For $\lambda_x, \lambda_v, \lambda_w>0$, we have	
\begin{enumerate}
\item First derivative:
\begin{align*}
\mathcal{E}_2'(t)=  -\lambda_v V_{\phi}(t).
\end{align*}
\item Second derivative:
\begin{align*}
\mathcal{E}_2''(t)= -\frac{\lambda_v}{N} \sum_{i=1}^{N}\sum_{j=1}^{N}  \bigg( \hat{\phi}'_{ij}|\hat{v}_i-\hat{v}_j|^2 + 2\hat{\phi}_{ij}\left(\hat{v}_i-\hat{v}_j\right) \cdot \left(\hat{v}'_i-\hat{v}'_j\right) \bigg).
\end{align*}
\item Third derivative:
\begin{align*}
 \mathcal{E}_2'''(t)
&= -\frac{\lambda_v}{N} \sum_{i=1}^{N}\sum_{j=1}^{N} \hat{\phi}''_{ij}|\hat{v}_i-\hat{v}_j|^2 \cr
&- \frac{\lambda_v}{N} \sum_{i=1}^{N}\sum_{j=1}^{N}  \bigg(  4\hat{\phi}'_{ij}\left(\hat{v}_i-\hat{v}_j\right) \cdot \left(\hat{v}'_i-\hat{v}'_j\right) \bigg)\cr
&- \frac{\lambda_v}{N} \sum_{i=1}^{N}\sum_{j=1}^{N}  \bigg(  2\hat{\phi}_{ij}\left(\hat{v}'_i-\hat{v}'_j\right) \cdot \left(\hat{v}'_i-\hat{v}'_j\right) \bigg)\cr
&- \frac{\lambda_v}{N} \sum_{i=1}^{N}\sum_{j=1}^{N}  \bigg(  2\hat{\phi}_{ij}\left(\hat{v}_i-\hat{v}_j\right) \cdot \left(\hat{v}''_i-\hat{v}''_j\right) \bigg).
\end{align*}	
\end{enumerate}
\end{lemma}
\begin{proof} (1) We differentiate $\mathcal{E}_2(t)$ and use Lemma \ref{new energy} to get
\begin{align*}
\frac{d}{dt}\bigg\{ \lambda_w X  +V + \frac{\lambda_x}{\beta_\gamma}S_\gamma\bigg\} &= 2 \lambda_w C(X,V) + \bigg\{-2\lambda_w C(X,V)
- \lambda_x C_{\phi}(X,V) - \lambda_v V_{\phi}\bigg\} \\
&+\frac{\lambda_x}{\beta_\gamma} \beta_\gamma C_{\phi}(X,V)    \cr
&=  -\lambda_v   V_{\phi}.
\end{align*}	
Identities in (2) and (3) follow directly from differentiating $-\lambda_vV_{\phi}$ in the r.h.s.	
\end{proof}

\begin{lemma}\label{equilibrium}
Fix  $t>0$. Assume  $\{\hat{v}_i(t)\}_{i=1,...,N} $ are all identical, but $\{\hat{x}_i(t)\}_{i=1,...,N} $ are not. That is,
\[
\sum_{i,j}|\hat{v}_i(t)-\hat{v}_j(t)|=0,\quad \sum_{i,j}|\hat{x}_i(t)-\hat{x}_j(t)|\neq0.
\]
Then the first, second derivatives of $\mathcal{E}_2$ vanish at $t$:
\[
\mathcal{E}'_2(t)=\mathcal{E}''_2(t)=0,
\]
and the third derivative of $\mathcal{E}_2$ is strictly negative at $t$:
$$
\mathcal{E}'''_2(t)<0.
$$
\end{lemma}
\begin{proof}
$\mathcal{E}'_2(t)=\mathcal{E}''_2(t)=0$ is clear from Lemma  $\ref{E_2'''}$ (1) and (2).
For $\mathcal{E}'''_2(t)$, recall from Lemma $\ref{E_2'''}$ (3) that, when $\{\hat{v}_i(t)\}_{i=1,...,N} $ are all identical, all but the third term in the r.h.s vanishes, yielding
\begin{align}\label{turnback}
\mathcal{E}_2'''(t)
&= - \frac{2\lambda_v}{N} \sum_{i=1}^{N}\sum_{j=1}^{N} \hat{\phi}_{ij}  \big|\hat{v}'_i-\hat{v}'_j\big|^2.
\end{align}
Therefore, our goal is reduced to showing that the right hand side is not zero.
For this, we note from (\ref{centralized herding model}) that
\begin{align*}
\hat{v}_i^{\prime}-\hat{v}_j^{\prime}&=
-\lambda_w\left(\hat{x}_i-\hat{x}_j\right)+\frac{\lambda_x}{N}\sum_{k=1}^N\hat{\phi}_{ik}(\hat{x}_k-\hat{x}_i)+\frac{\lambda_v}{N}\sum_{k=1}^N\hat{\phi}_{ik}(\hat{v}_k-\hat{v}_i)\cr
&-\frac{\lambda_x}{N}\sum_{k=1}^N\hat{\phi}_{jk}(\hat{x}_k-\hat{x}_j)-\frac{\lambda_v}{N}\sum_{k=1}^N\hat{\phi}_{jk}(\hat{v}_k-\hat{v}_j),
\end{align*}
which, under our assumption of identical favorability, reduces to
\begin{align*}
\hat{v}'_{i}-\hat{v}'_{j}&=
-\lambda_w\left(\hat{x}_{i}-\hat{x}_{j}\right)+\frac{\lambda_x}{N}\sum_{k=1}^N\hat{\phi}_{ik}(\hat{x}_k-\hat{x}_{i})-\frac{\lambda_x}{N}\sum_{k=1}^N\hat{\phi}_{jk}(\hat{x}_k-\hat{x}_{j}).
\end{align*}
Now, for a vector $\hat{x}_i$, let $\hat{x}_i^k$ denote the $k$th element of $\hat{x}_i$. That is, $\hat{x}_i=(\hat{x}_i^1,\hat{x}_i^2,\dots,\hat{x}_i^M)$. Then from our assumption, we can find $i_0 \ne j_0$ and $d$ such that $$\min_{i} \hat{x}_i^d=\hat{x}_{i_0}^d < \hat{x}_{j_0}^d = \max_{i} \hat{x}_i^d \quad \text{and}\quad
\hat{x}^{d}_{i_0} \leq \cdots \leq \hat{x}^{d}_{j_0}.
$$
For such choice of $i_0, j_0$ and $d$, we have
\begin{align*}
\hat{v}'^{d}_{i_0}-\hat{v}'^{d}_{j_0}&=
-\lambda_w\left(\hat{x}^{d}_{i_0}-\hat{x}^{d}_{j_0}\right)+\frac{\lambda_x}{N}\sum_{k=1}^N\hat{\phi}_{i_0k}(\hat{x}^{d}_k-\hat{x}^{d}_{i_0})-\frac{\lambda_x}{N}
\sum_{k=1}^N\hat{\phi}_{j_0k}(\hat{x}^{d}_k-\hat{x}^{d}_{j_0}) >0.
\end{align*}
We now turn back to (\ref{turnback}) with this observation to obtain the desired result:
\begin{align*}
\mathcal{E}_2'''(t)
&= - \frac{2\lambda_v}{N} \sum_{i=1}^{N}\sum_{j=1}^{N} \hat{\phi}_{ij}  \big|\hat{v}'_i-\hat{v}'_j\big|^2\cr
&\leq - \frac{2\lambda_v}{N} \sum_{i=1}^{N}\sum_{j=1}^{N} \hat{\phi}_{ij}  \big|\hat{v}'^{d}_i-\hat{v}'^{d}_j\big|^2\cr
&\leq - \frac{2\lambda_v}{N}  \hat{\phi}_{i_0j_0}  \big|\hat{v}'^{d}_{i_0}-\hat{v}'^{d}_{j_0}\big|^2\cr
&<0.
\end{align*}
\end{proof}	
\subsection{Proof of Theorem \ref{main theorem 2}:}
We first recall the following invariance principle by LaSalle \cite{L}:
\begin{definition}\label{LaSalle's therem}{\bf\cite{L}}
A set $\mathcal{M}$ is said to be invariant if each solution starting
in $\mathcal{M}$ remains in $\mathcal{M}$ for all $t$. That is,
$$x(0)\in M \Rightarrow x(t)\in \mathcal{M} \text{ for all }t.$$
\end{definition}

\begin{theorem}\label{LaSalle's therem}{\bf\cite{L}}
Consider the system of differential equations
\begin{equation}\label{prb}
\dot{x}=F(x)
\end{equation}
where $x(t)=\big(x_1(t),\dots,x_N(t)\big)$ and $F$ is a vector field.
Let $L(x)$ be a scalar function with continuous first partials for all $x$. Assume that
\begin{align*}
&i) \ L(x(t)) > 0 \text{ for all } x \neq 0,\\
&ii) \ \dot{L}(x(t)) \leq 0 \text{ for all } x.
\end{align*}	
Let $E$ be the set of all points where $ \dot{L}(x) = 0$, and let $\mathcal{M}$ be the largest invariant set contained in E. Then every solution
of (\ref{prb}) bounded for all $t \geq 0$ approaches $\mathcal{M}$ as $t \rightarrow \infty$.	
\end{theorem}

We now start the proof of Theorem  \ref{main theorem 2}.  First, recall that the herding functional $\mathcal{E}_2(t)$ is non-negative:
$$
\mathcal{E}_2(t) \geq 0,
$$
and vanishes only when $\hat{x}(t) = \hat{v}(t) = 0$.
We also have
\begin{equation}\label{boundness}
	\mathcal{E}_2'(t) = -\lambda_v V_{\phi}(t) \leq 0.
\end{equation}
The equality holds only when $\{\hat{v}_i\}_{i=1,...,N}$ are all identical. (\ref{boundness}) also implies the boundness of the solutions.
Therefore, $\mathcal{E}_2$ satisfies the conditions of Theorem \ref{LaSalle's therem}.
Now, define $E$ to be the null-space of $V_{\phi}(t)$:
$$E:= \big\{(\hat{x}_1,\dots,\hat{x}_N,\hat{v}_1,\dots,\hat{v}_N
) \in \mathbb{R}^{2MN} \ | \ \hat{v}_1=\cdots = \hat{v}_N \big\}.$$
Since we are considering the centralized model, and $\hat{v}_c(t)=0$, $E$ actually is
$$E:= \big\{(\hat{x}_1,\dots,\hat{x}_N,\hat{v}_1,\dots,\hat{v}_N
) \in \mathbb{R}^{2MN} \ | \ \hat{v}_1=\cdots = \hat{v}_N=0 \big\}.$$
 Let
 $\mathcal{M}$ be the largest invariance set in $E$. In view of the above invariance theorem, our goal is to verify that $\mathcal{M}$ is trivial:
\begin{align}\label{MM}
\mathcal{M}= \big\{0\big\} \subseteq \mathbb{R}^{2MN} .
\end{align}
For this, suppose contrarily that there exists  an open interval $I$ and a solution $\{\hat{x}_i(t),\hat{v}_i(t)\}_{i=1}^N$  to (\ref{centralized herding model}) in $\mathcal{M}$ such that
\[
\sum_{i,j}|\hat{x}_i(s)-\hat{x}_j(s)|\neq0 \mbox{ for } s\in I.
\]
Since $\{\hat{x}_i(t),\hat{v}_i(t)\}_{i=1}^N\in E$ by definition, we have
\[
\hat{v}_1(t)=\dots=\hat{v}_N(t)=0 ,
\]
and by Lemma \ref{equilibrium}, the first and second derivative of $\mathcal{E}_2(t)$ vanishes while
the third derivative $\mathcal{E}_2'''(t)$ remains strictly negative on $I$:
\[
\mathcal{E}_2'(t)=\mathcal{E}_2''(t)=0,\quad \mathcal{E}_2'''(t)<0\mbox{ on }I.
\]
Therefore, for any $[t_1,t_2]\subset I$, we have from the Taylor's theorem
\begin{align*}
0=\mathcal{E}_2^{\prime}(t_2)&=\mathcal{E}_2^{\prime}(t_1)+(t_2-t_1)\mathcal{E}_2^{\prime\prime}(t_1)+\int^{t_2}_{t_1}(t_2-s)^2\mathcal{E}_2^{\prime\prime\prime}(s)ds\cr
&=\int^{t_2}_{t_1}(t_2-s)^2\mathcal{E}_2^{\prime\prime\prime}(s)ds\cr
&<0,
\end{align*}
 which is a contradiction. This proves (\ref{MM}). The desired result then follows from Theorem \ref{LaSalle's therem}.

\section{Numerical Simulation}	

 In this section, we present three numerical tests demonstrating the herding behavior in the market.
In Test 1, we numerically  verify Theorem \ref{main theorem 1} and Theorem \ref{main theorem 2} in the case
($M=1$, $N=5$) with different choices of parameters and initial data.
In Test 2, we present trajectories of the numerical solution to ($\ref{centralized herding model}$) for two dimensional problem ($M=2, N=4$) to visualize the herding phenomenon in multi-d case.
In Test 3, we give two dimensional histograms for each variable $x$ and $v$ with $M=2, N=500$ for the simulation of large number of players. We employ a fourth-order Runge-Kutta method for the time evolution with fixed time steps in all simulations.

\subsection{Numerical test 1} We recall three scenarios of the dynamic market signals presented in Section 2 and fix
throughout test 1 as
$$w_1(x,t)=4\cos(4t), \quad w_2(x,t)=x_1(t), \quad w_3(x,t)=\frac{1}{N}\sum_{i=1}^{N}x_i(t).$$
For simplicity, we consider five market players $N=5$ and one asset $M=1$.

In the numerical  Test 1-1 and 1-2 below, the initial data $(x_i(0), v_i(0))$ have been chosen randomly,
and the only difference is the choice of $x_1(0)$:
\[
\mbox{Test 1-1}:~x_1(0)=10,\qquad \mbox{Test 1-2}:~x_1(0)=-10.
\]
The choice $x_1(0)=10$ means that initially market player $x_1$ has viewed the market positively,
and choice $x_1(0)=-10$ means the opposite.
This is to compare the influence of the assessment of the influential player  $x_1(t)$ on the market.

In Test 1-3, we consider the dynamics of solutions corresponding to an initial configuration that doesn't satisfy the conditions of Theorem \ref{main theorem 1}.
We observe that essential features of Theorem \ref{main theorem 1} break down but the result of Theorem \ref{main theorem 2} still holds. It demonstrates that the conditions in Theorem \ref{main theorem 1} are essential.  (See Remark \ref{rmk2}.)
Throughout Test 1-1 to 1-3, we consider numerical solutions to the non-centralized model (\ref{herding model}) to clearly manifest the influence of $(x_c,v_c)$
on the herding dynamics.
\subsubsection{\bf Numerical test 1-1:} {\bf (The case of $x_1(0)=10$)}\label{NT1-1}}
In this test, we choose $\gamma=1.5$, $\lambda_x=0.1$, $\lambda_v=3$ and $\lambda_w=2$, which satisfy the parameter assumption (\ref{parameter assumption}). Initial data is randomly chosen in $[-10,10]\times[-10,10]$ to satisfy (\ref{assumption}) and (\ref{assumption2}) in Theorem \ref{main theorem 1}. We obtain numerical solutions to the original non-centralized herding model (\ref{herding model}) corresponding to each of the three scenarios and plot them in Figure \ref{fig:test11w1},\ref{fig:test11w2} and \ref{fig:test11w3}, respectively. The final time is taken as $T=6$.



\begin{figure}[!t]
	\includegraphics[width=0.45\linewidth]{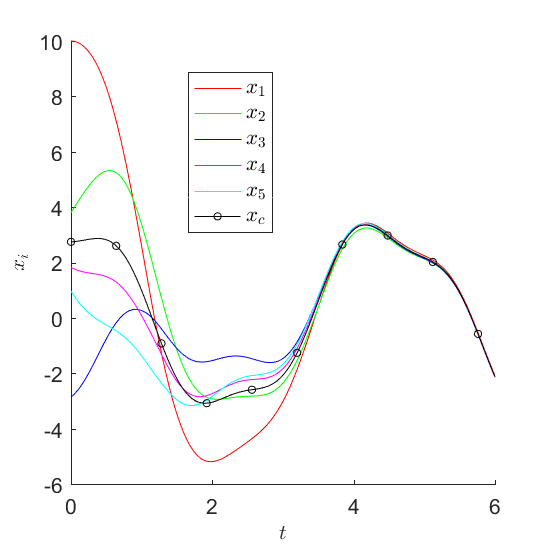}
	\includegraphics[width=0.45\linewidth]{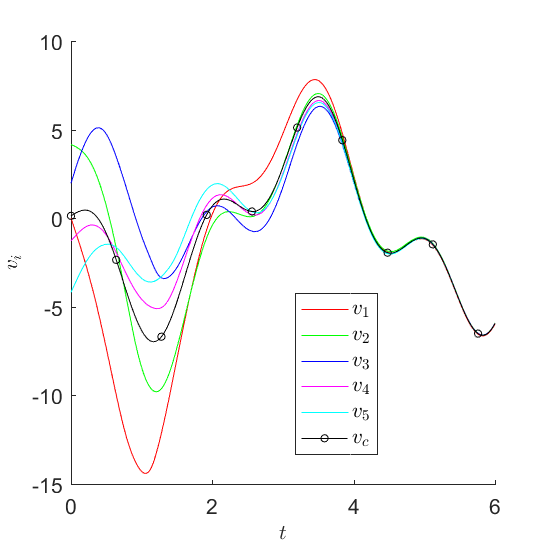}
	\caption{$w_1(x,t)=4\cos(4t)$}
	\label{fig:test11w1}
\end{figure}

\begin{figure}[!h]
	\includegraphics[width=0.45\linewidth]{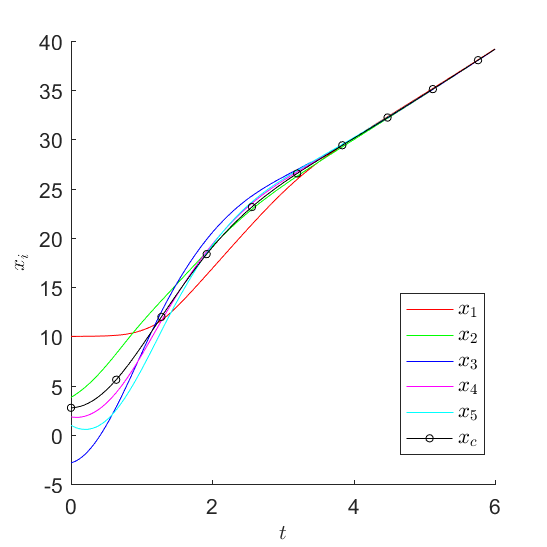}
	\includegraphics[width=0.45\linewidth]{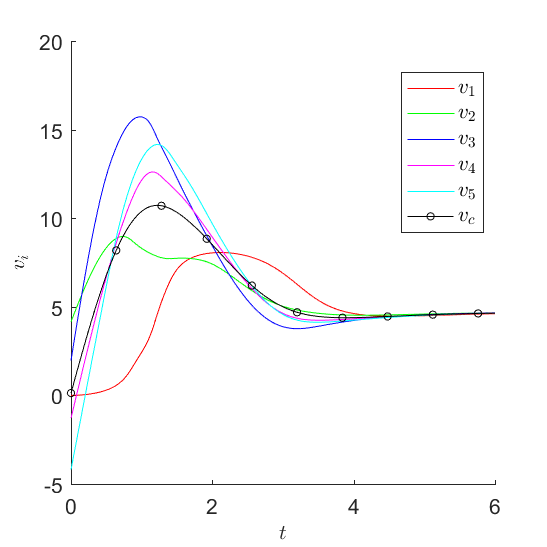}
	\caption{$w_2(x,t)=x_1(t)$}
	\label{fig:test11w2}
\end{figure}

In Figure \ref{fig:test11w1}, we plost $x$ and $v$ in the case of $w_1(x,t)=4\cos(4t)$.
We observe that  $x_i,v_i$ tend to show similar patterns after $t=4$. We also see that the average favorability $v_c$ moves up and down. Since the market signal $w_1$ has  a oscillating pattern,  it is natural to observe such situation in which the market players tend to change their minds frequently.

In Figure \ref{fig:test11w2}, the behaviors of $x$ and $v$ are provided when $w_2(x,t)=x_1$.
We see that the numerical solution to the second scenario shows very different phase compared to
those presented in  Figure \ref{fig:test11w1} even though they start with the same initial random data.
After $t=4$, each $v_i$ seems to converge to a fixed constant about $4$, which leads to linear increase in $x_i$.

In Figure \ref{fig:test11w3}, behaviors of $x,v$ are shown for $w_3(x,t)=\frac{1}{N}\sum_{i=1}^{N}x_i(t)$.
We observe that $v_c$ remains unchanged.
Recalling Lemma \ref{average} implies $\frac{d}{dt}v_c(t)=0$ in this case, it is natural to have such a fixed constant $v_c$.
We can observe that it also shows the collective behaviors after $t=4$.
\begin{figure}[!t]
	\includegraphics[width=0.45\linewidth]{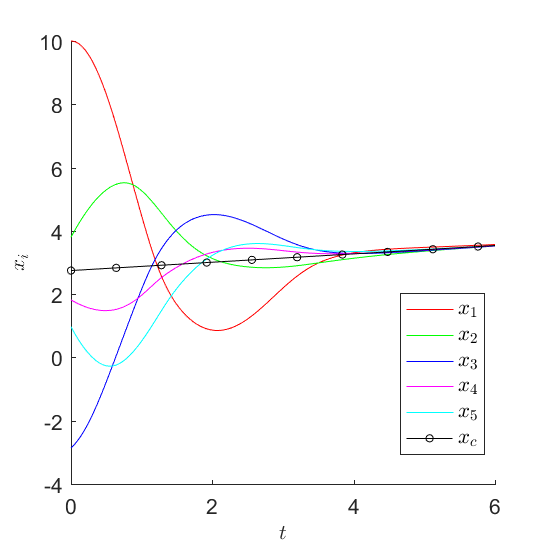}
	\includegraphics[width=0.45\linewidth]{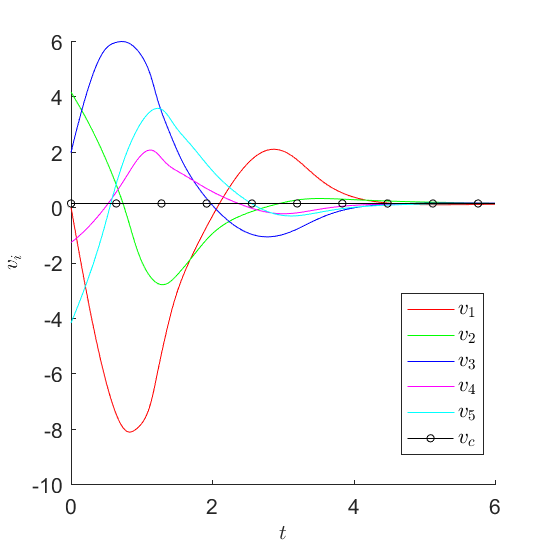}
	\caption{$w_3(x,t)=\frac{1}{N}\sum_{i=1}^{N}x_i(t)$}
	\label{fig:test11w3}
\end{figure}

\begin{figure}[!h]
	\centering
	\begin{subfigure}[b]{0.328\linewidth}
		\includegraphics[width=1\linewidth]{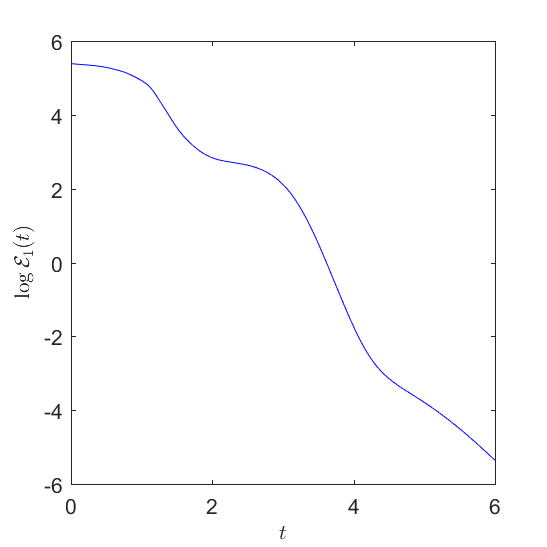}
		\subcaption{Herding energy $\mathcal{E}_1(t)$}
		\label{fig:test11E1}
	\end{subfigure}	
	\begin{subfigure}[b]{0.328\linewidth}
		\includegraphics[width=1\linewidth]{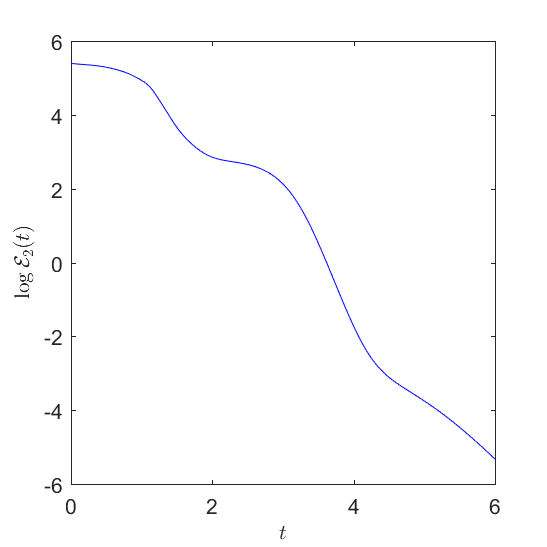}
		\subcaption{Herding energy $\mathcal{E}_2(t)$}
		\label{fig:test11E2}	
	\end{subfigure}
	\begin{subfigure}[b]{0.328\linewidth}
		\includegraphics[width=1\linewidth]{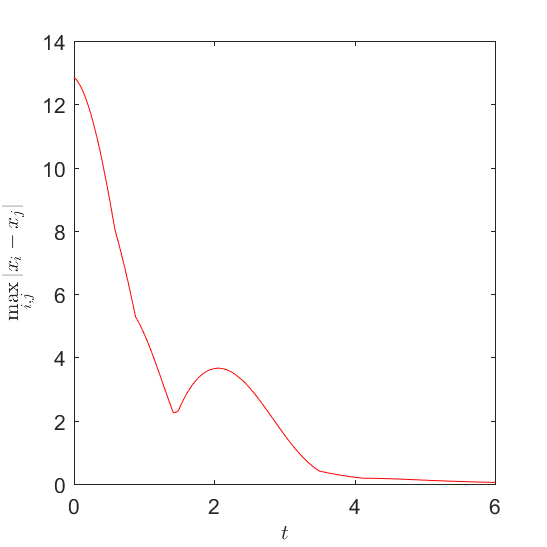}	
		\caption{  $\max_{i,j}|x_i-x_j|$}
		\label{fig:test11boundx}	
	\end{subfigure}
	\caption{Herding energies and bound}	
\end{figure}

\vspace{2mm}

We now check the evolution of two herding energies: $\mathcal{E}_1(t)$ and $\mathcal{E}_2(t)$.
 In Figure \ref{fig:test11E1} and \ref{fig:test11E2}, we plot $\mathcal{E}_1(t)$ and $\mathcal{E}_2(t)$
 in the log sense, respectively. We observe that each herding energy decreases monotonically toward $0$.
In Figure \ref{fig:test11boundx}, it is verified that  $\max_{i,j}|x_i(t)-x_j(t)|$ can not attain the bound $2\max_{i,j}|x_i(0)-x_j(0)|$ for all $t \geq 0$ as was guaranteed in (\ref{2bae}).


\subsubsection{\bf Numerical test 1-2:} \label{NT1-2}{\bf (The case of $\bf{x_1(0)=-10})$}
In this test, we replace $x_1(0)$ as
 $$x_1(0)=10 \rightarrow x_1(0)=-10,$$
 with all the other initial configuration and parameters fixed. Note that the new initial condition also satisfies conditions of Theorem \ref{main theorem 1}.

In Figure \ref{fig:test12w1}, \ref{fig:test12w2} and \ref{fig:test12w3}, we plot the numerical solutions to (\ref{herding model}) for  each scenario upto $T=6$, respectively. In all these figures, we observe the herding phenomenon in the expected rate of returns and favorabilities.

\begin{figure}[h]
	\includegraphics[width=0.45\linewidth]{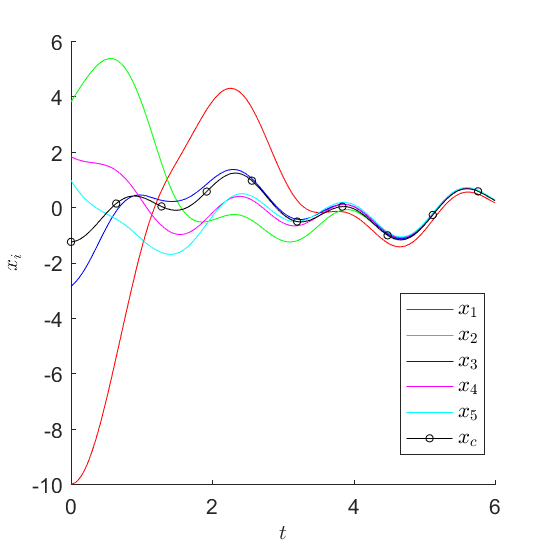}
	\includegraphics[width=0.45\linewidth]{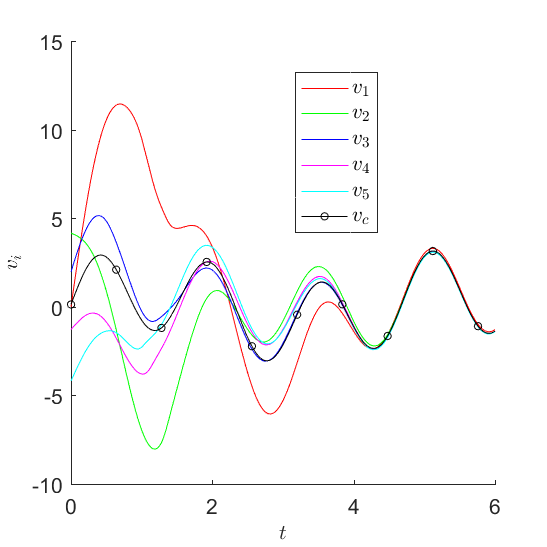}
	\caption{$w_1(x,t)=4\cos(4t)$}
	\label{fig:test12w1}
\end{figure}

\begin{figure}[!h]
	\includegraphics[width=0.45\linewidth]{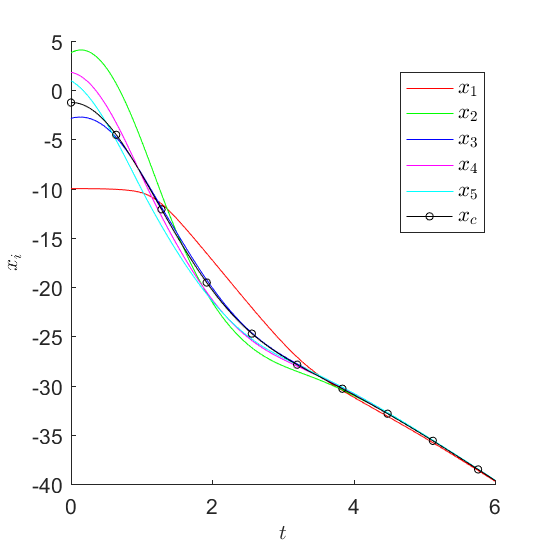}
	\includegraphics[width=0.45\linewidth]{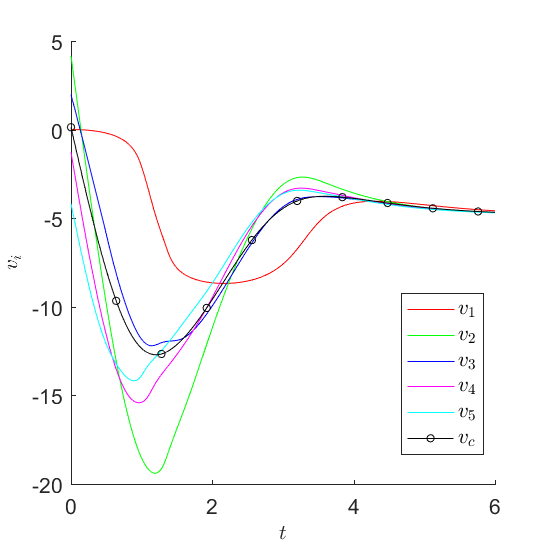}
	\caption{$w_2(x,t)=x_1(t)$}
	\label{fig:test12w2}
\end{figure}
\begin{figure}[!h]
	\includegraphics[width=0.45\linewidth]{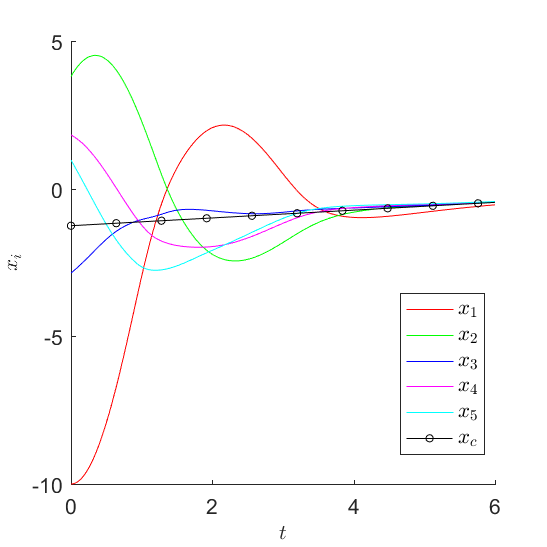}
	\includegraphics[width=0.45\linewidth]{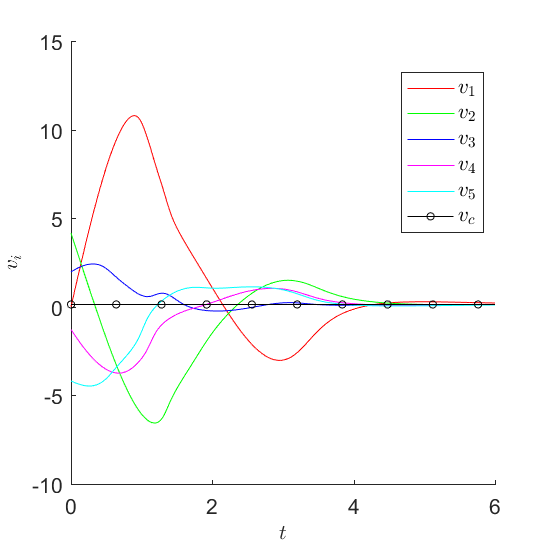}
	\caption{$w_3(x,t)=\frac{1}{N}\sum_{i=1}^{N}x_i(t)$}
	\label{fig:test12w3}
\end{figure}

\begin{figure}[!h]
	\centering
	\begin{subfigure}[b]{0.328\linewidth}
		\includegraphics[width=1\linewidth]{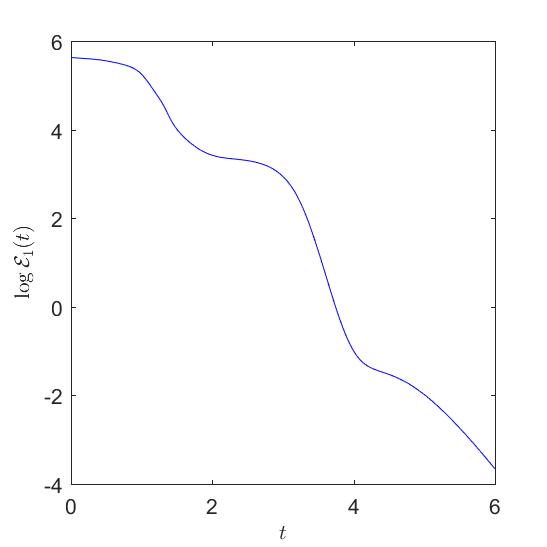}	
		\subcaption{Herding energy $\mathcal{E}_1(t)$}
		\label{fig:test12E1}
	\end{subfigure}	
	\begin{subfigure}[b]{0.328\linewidth}
		\includegraphics[width=1\linewidth]{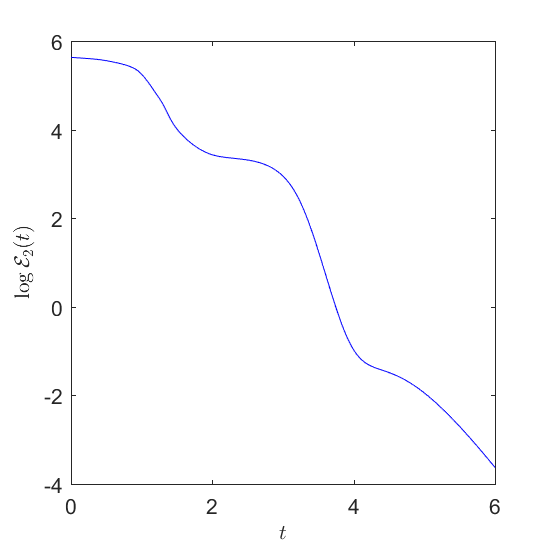}	
		\subcaption{Herding energy $\mathcal{E}_2(t)$}
		\label{fig:test12E2}	
	\end{subfigure}
	\begin{subfigure}[b]{0.328\linewidth}
		\includegraphics[width=1\linewidth]{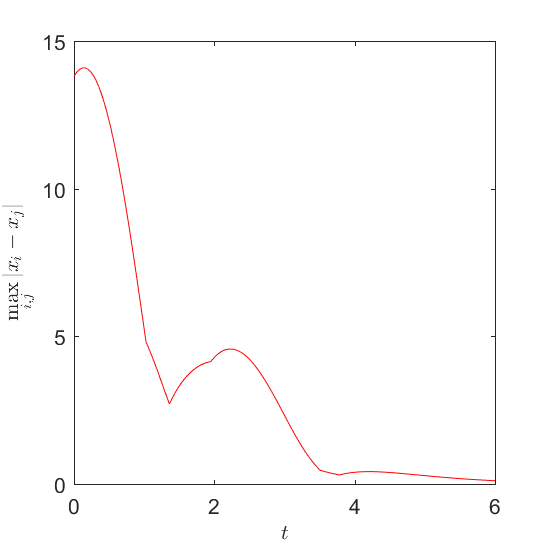}	
		\caption{ $\max_{i,j}|x_i-x_j|$}
		\label{fig:test12boundx}	
	\end{subfigure}
	\caption{Herding energies and bound}	
\end{figure}
The main difference compared to the Test 1-1 is observed in Figure \ref{fig:test12w2}, where we clearly see the influence of negative evaluation on the asset by player $x_1$ in that, in contrast to Figure \ref{fig:test11w2}, the negative assessment of the influential market player on the rate of return of the asset $(x_1(0)<0)$  leads to the ever-decreasing expected rate of returns, and emergence of negative value of $v_c$ even if we initially have $v_c(0)>0$.

In Figure \ref{fig:test12E1}, Figure \ref{fig:test12E2} and Figure \ref{fig:test12boundx}, we can see that each herding energy decays monotonically and the uniform bound $2\max_{i,j}|x_i(0)-x_j(0)|$ can not be attained for all $t \geq 0$, as in the test 1-1.


\subsubsection{Numerical test 1-3}{\bf (Removal of restrictions on parameters)} \label{NT1-3}
In this test, we provide a numerical example which shows that   the herding behavior still occurs
even when the restrictions on the paramenters and initial configurations imposed on Theorem \ref{main theorem 1} are not satisfied, as was guranteed by Theorem \ref{main theorem 2}.


 For this, we choose $\gamma=1$, $\lambda_x=2$, $\lambda_v=1$ and $\lambda_w=0.5$. Since there is no restriction on  initial data in Theorem \ref{main theorem 2}, we choose them randomly in $[-15,15]\times[-15,15]$. In Figure \ref{fig:test13w1} - \ref{fig:test13w3}, we plot the numerical solutions to (\ref{herding model}) upto $T=15$. These figures show that the
 desired herding phenomena occurs.

We, however, observe that various important features of Theorem \ref{main theorem 1} do not hold.
Figure \ref{fig:test13E1} shows that $\mathcal{E}_1(t)$ does not decrease monotonically anymore, and even may exceed $\mathcal{E}_1(0)$. We observe that $\mathcal{E}_1(t)$ can take  negative values. In Figure \ref{fig:test13boundx}, we also see that
$\max_{i,j}|x_i(t)-x_j(t)|$ can exceed  $2\max_{i,j}|x_i(0)-x_j(0)|$ in this case,  which implies that the
proof of Theorem \ref{main theorem 1} in Section 4 breaks down.
Meanwhile, in Figure \ref{fig:test13E2}, $\mathcal{E}_2(t)$ decreases monotonically toward $0$ even with same initial data.

\begin{figure}[h]
	\includegraphics[width=0.45\linewidth]{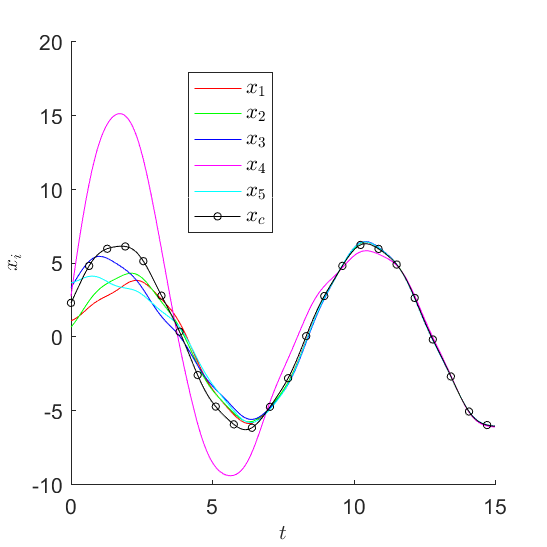}
	\includegraphics[width=0.45\linewidth]{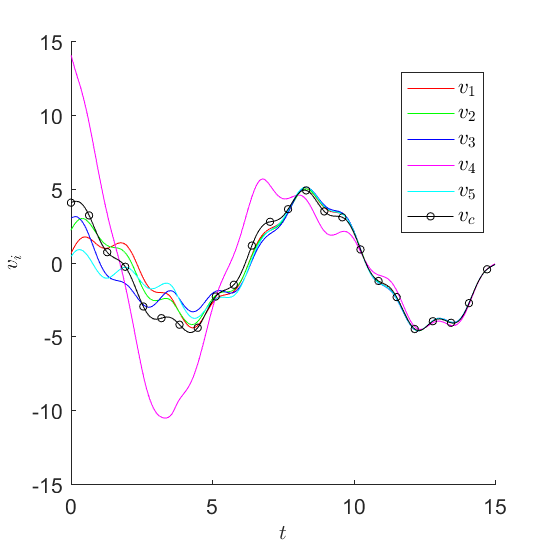}
	\caption{$w_1(x,t)=4\cos(4t)$}
	\label{fig:test13w1}
\end{figure}
\begin{figure}[!h]
	\includegraphics[width=0.45\linewidth]{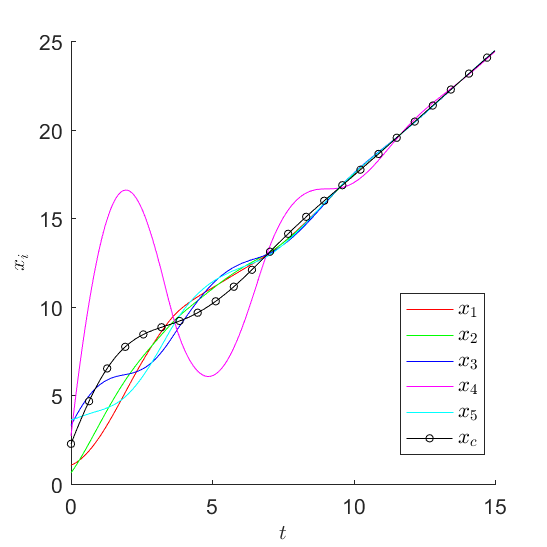}
	\includegraphics[width=0.45\linewidth]{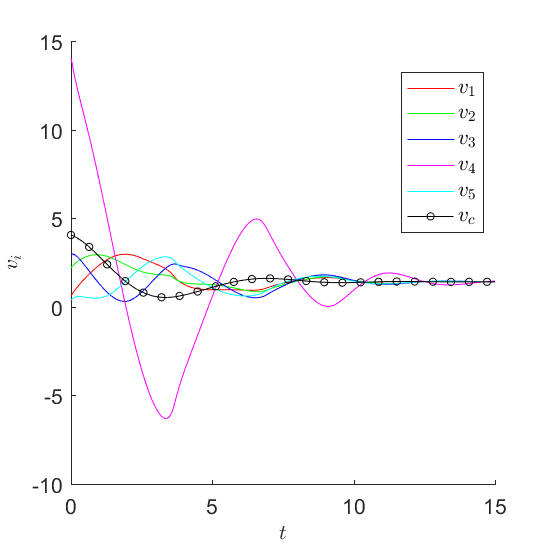}
	\caption{$w_2(x,t)=x_1(t)$}
	\label{fig:test13w2}
\end{figure}
\begin{figure}[!h]
	\includegraphics[width=0.45\linewidth]{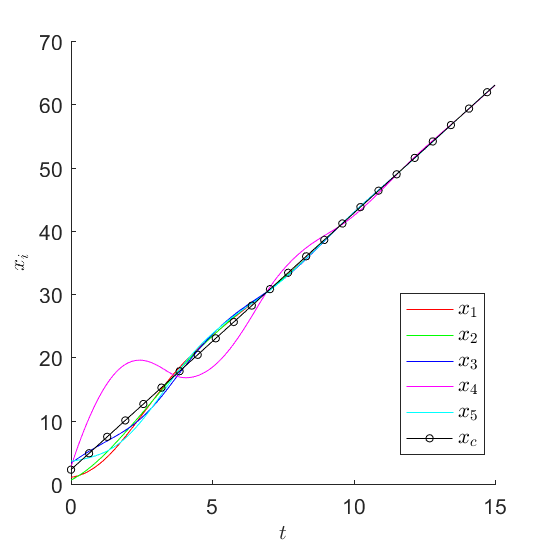}
	\includegraphics[width=0.45\linewidth]{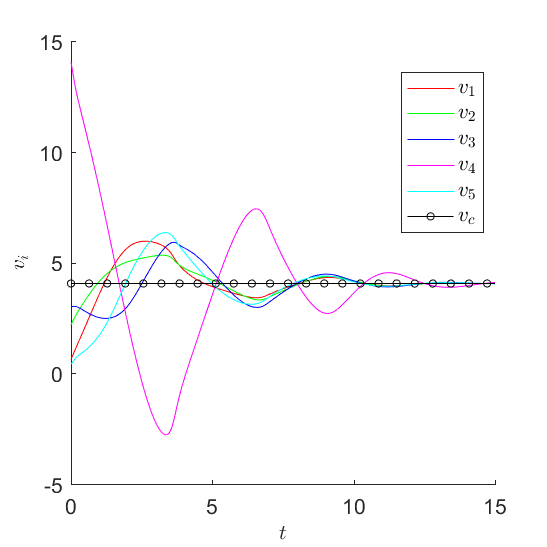}
	\caption{$w_3(x,t)=\frac{1}{N}\sum_{i=1}^{N}x_i(t)$}
	\label{fig:test13w3}
\end{figure}

\begin{figure}[!]
	\centering
	\begin{subfigure}[b]{0.328\linewidth}
		\includegraphics[width=1\linewidth]{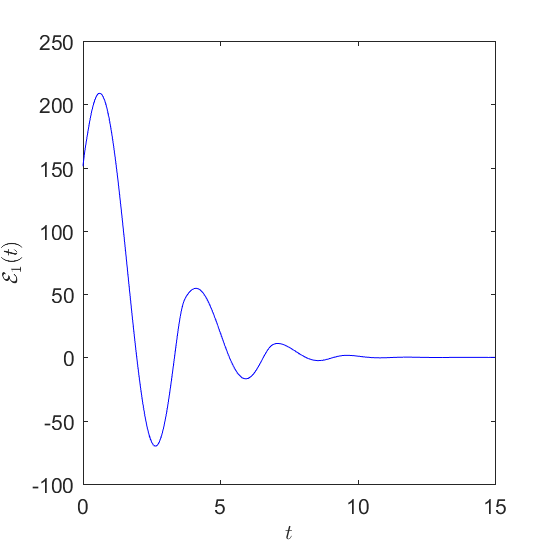}	
		\subcaption{Herding energy $\mathcal{E}_1(t)$}
		\label{fig:test13E1}
	\end{subfigure}	
	\begin{subfigure}[b]{0.328\linewidth}
		\includegraphics[width=1\linewidth]{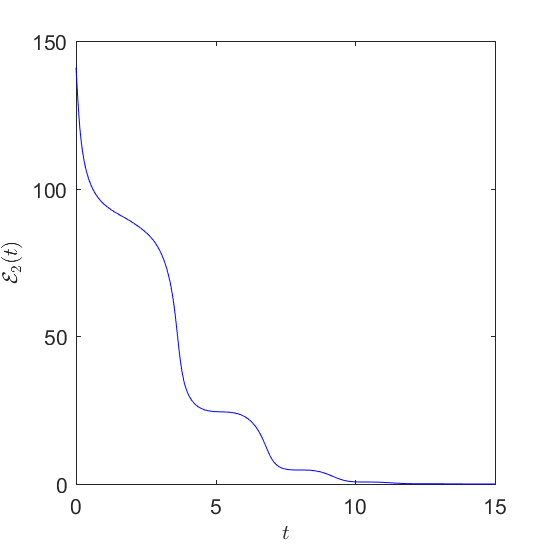}	
		\subcaption{Herding energy $\mathcal{E}_2(t)$}
		\label{fig:test13E2}	
	\end{subfigure}
	\begin{subfigure}[b]{0.328\linewidth}
		\includegraphics[width=1\linewidth]{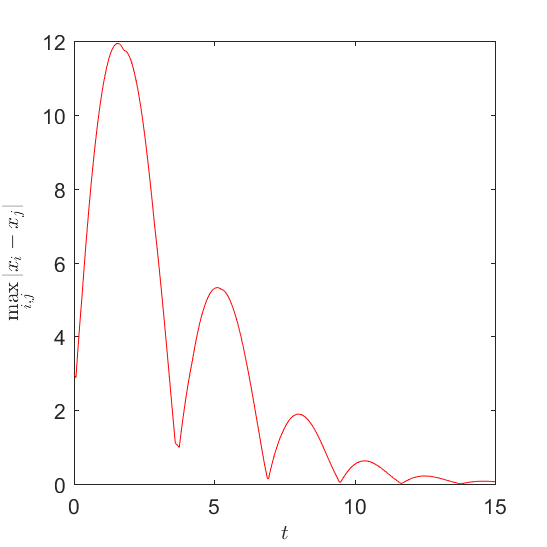}	
		\caption{ $\max_{i,j}|x_i-x_j|$ }
		\label{fig:test13boundx}	
	\end{subfigure}
\caption{Herding energies and bound}
\end{figure}

\subsection{Numerical test 2:} In this test, we consider the trajectory of numerical solutions to
the centralized herding model (\ref{centralized herding model}) to visualize the dynamics of solutions.
We choose $\gamma=2$, $\lambda_x=1$, $\lambda_v=1$ and $\lambda_w=1$, which do not fit into the condition of Theorem \ref{main theorem 1}. The herding phenomena is still guaranteed by Theorem \ref{main theorem 2}. For the clarity of simulation, we take the number of assets and the number of players as $M=2,~N=4$, and pick initial data from $[-5,5] \times [-5,5]$ such that $$x_c(0)=v_c(0)=0.$$

Figure \ref{fig:test20to5xv} -  \ref{fig:test220to25xv}  present the trajectory of each $x_i$ and $v_i$ upto final time $T=25$. In each figure, `o' and `x' stand for the endpoint and the starting point of each trajectory respectively.
These figures show how the configurations of $x_i$ and $v_i$ evolve as time passes.
Even though their movement seems to be very tangle, all fluctuations are getting smaller in each step, eventually leading to herding phenomena.

\begin{figure}[]
	\includegraphics[width=0.45\linewidth]{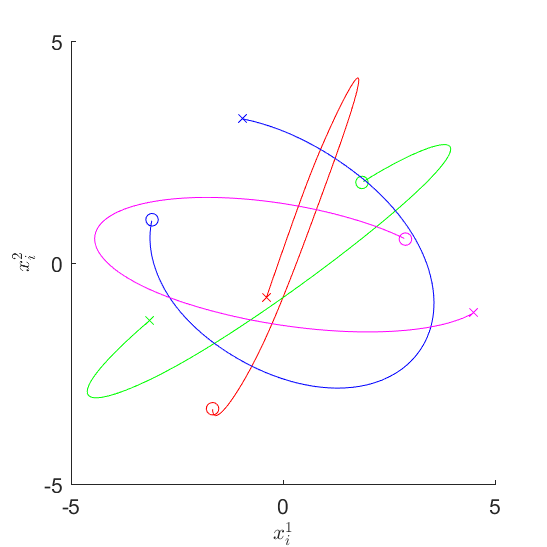}
	\includegraphics[width=0.45\linewidth]{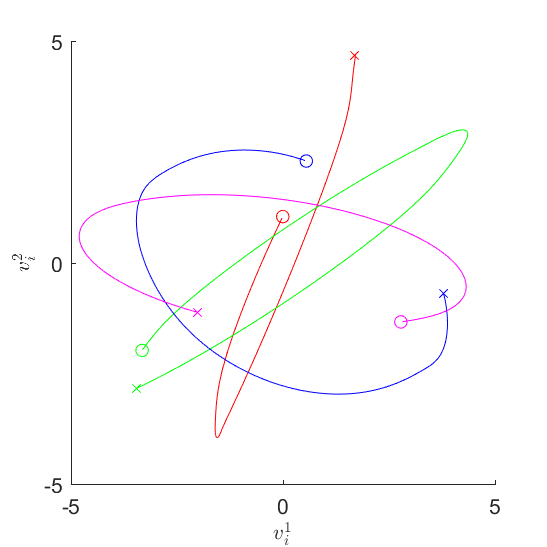}
	\caption{$t=0$ to $5$}
	\label{fig:test20to5xv}
	\includegraphics[width=0.45\linewidth]{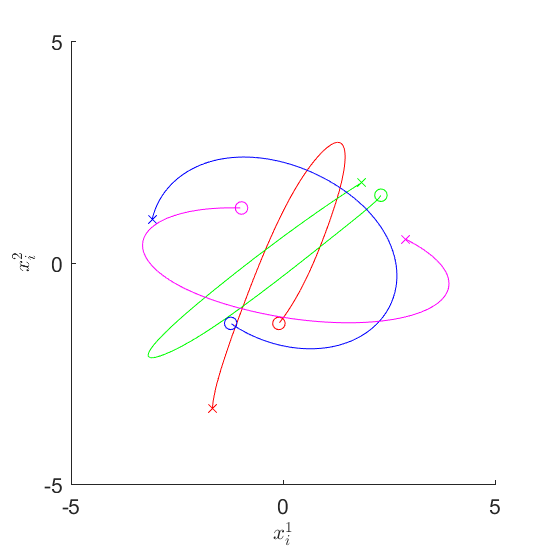}
	\includegraphics[width=0.45\linewidth]{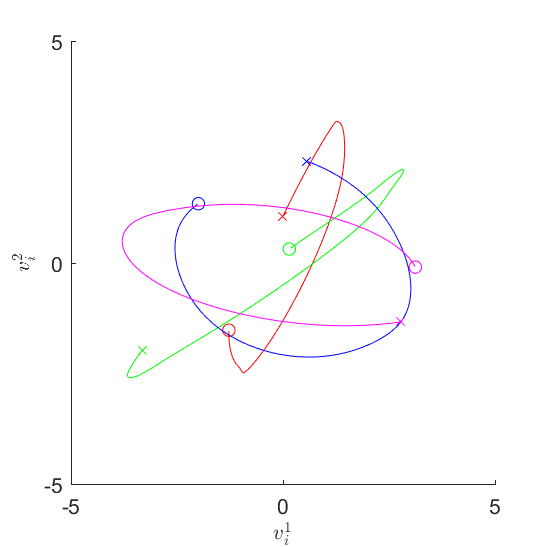}
	\caption{$t=5$ to $10$}
	\label{fig:test25to10xv}
	\includegraphics[width=0.45\linewidth]{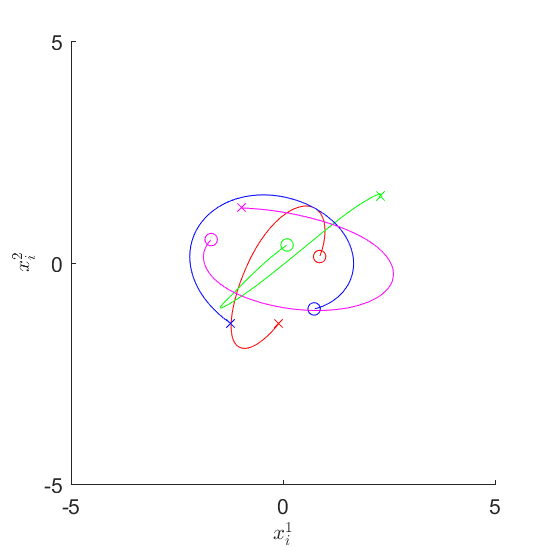}
	\includegraphics[width=0.45\linewidth]{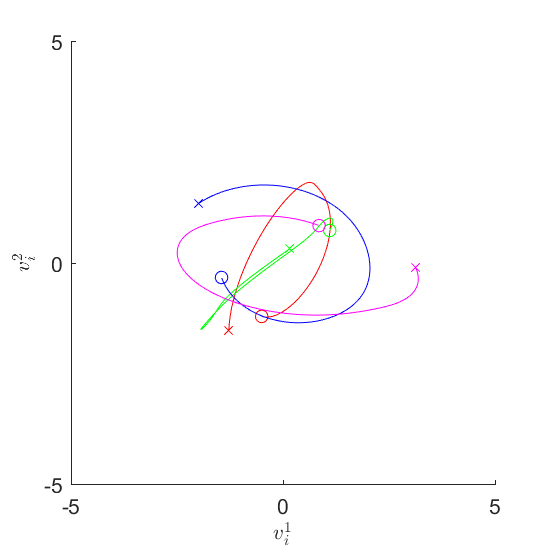}
	\caption{$t=10$ to $15$}
	\label{fig:test210to15xv}	
\end{figure}

\begin{figure}[]
	\centering
	\includegraphics[width=0.45\linewidth]{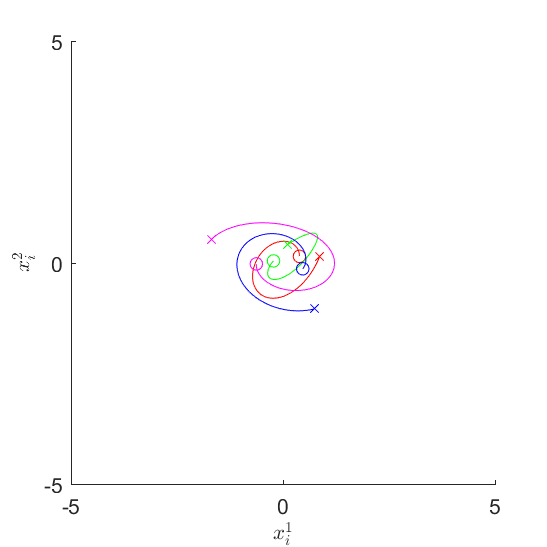}
	\includegraphics[width=0.45\linewidth]{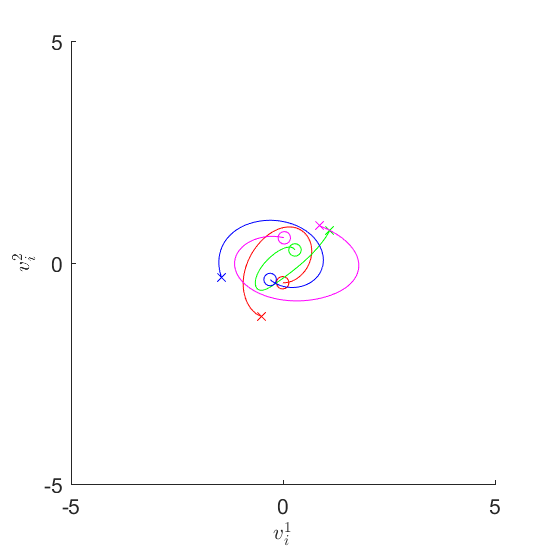}
	\caption{$t=15$ to $20$}
	\label{fig:test215to20xv}
	\includegraphics[width=0.45\linewidth]{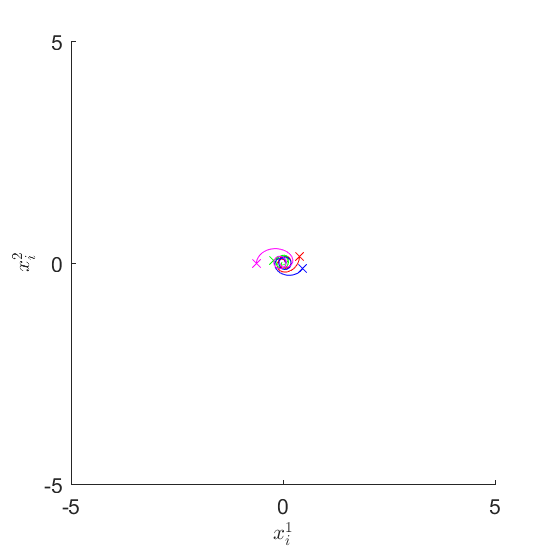}	
	\includegraphics[width=0.45\linewidth]{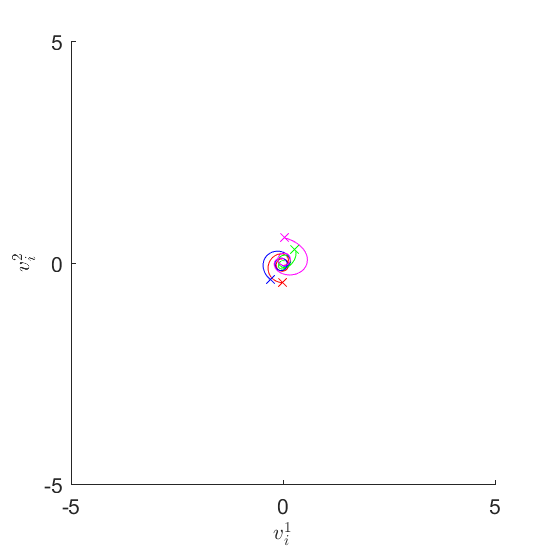}
	\caption{$t=20$ to $25$}
	\label{fig:test220to25xv}	
\end{figure}

\subsection{Numerical test 3:} In this test, we simulate the behavior of large number of players dealing with two assets $(M=2, N=500)$. We set $\gamma=1.5$, $\lambda_x=1$, $\lambda_v=1$ and $\lambda_w=1$ and choose initial data from $[-5,5]\times[-5,5]$ to be $$x_c(0)=v_c(0)=0.$$  We take final time as $T=20$. As in test 2, we consider the numerical solution to the centralized model (\ref{centralized herding model}).

In Figure \ref{fig:test3initial} - \ref{fig:test3sec20}, we plot histograms of $x_i$ and $v_i$ at $t=0, 5, 10, 15$ and $20$, respectively. In each histogram, each boundary cell count the number of outliers in its designated region. For example, the top left cell counts the number of market players in $(-\infty,-5) \times (5,\infty)$. We see that, except for Figure \ref{fig:test3initial}, there is no outlier at $t=5,10,15$ and $20$. As time flows, $x_i$ and $v_i$ are
concentrated to centers $(0,0)$ and $(0,0)$, respectively. In Figure \ref{fig:test3sec20}, we  observe that all market players have similar expected rate of returns and favorabilies enough to be contained in one pixel.

\begin{figure}[]
	\includegraphics[width=0.45\linewidth]{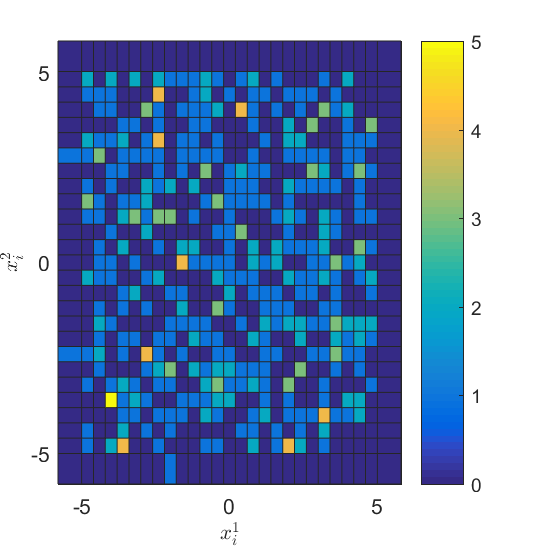}
	\includegraphics[width=0.45\linewidth]{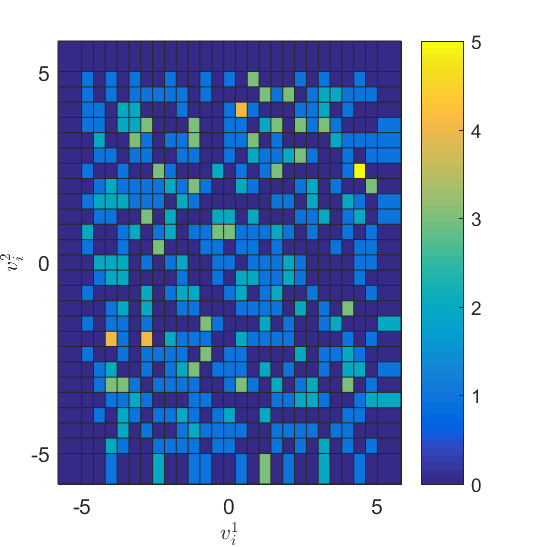}
	\caption{$t=0$}
	\label{fig:test3initial}
	\includegraphics[width=0.45\linewidth]{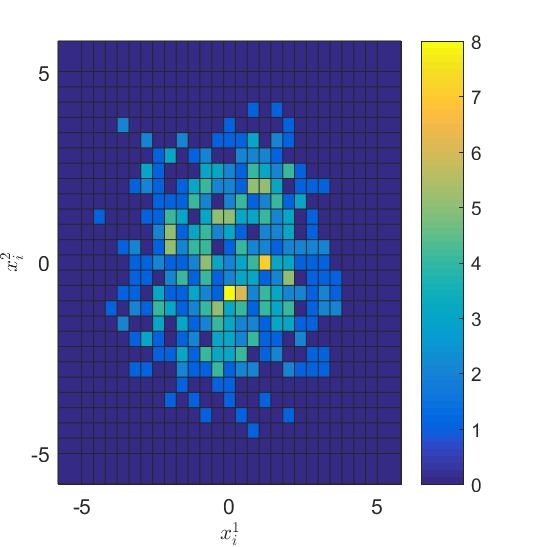}
	\includegraphics[width=0.45\linewidth]{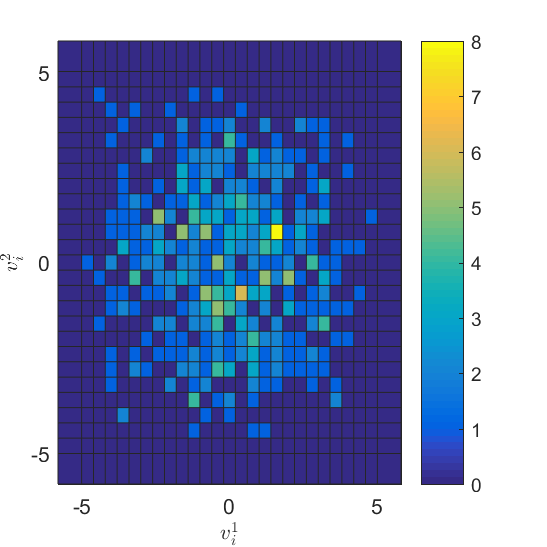}
	\caption{$t=5$}
	\label{fig:test3sec5}
	\includegraphics[width=0.45\linewidth]{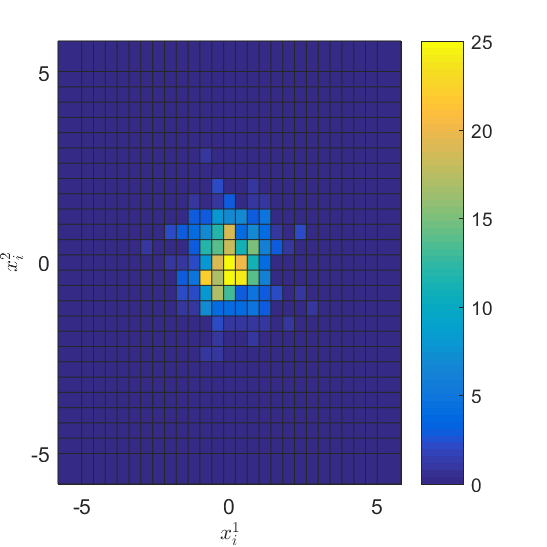}
	\includegraphics[width=0.45\linewidth]{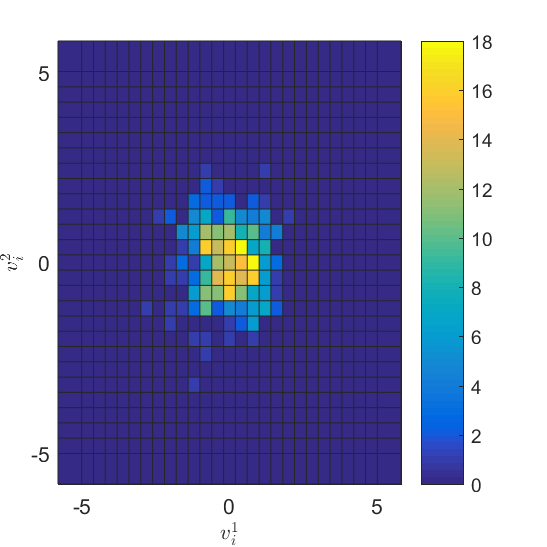}
	\caption{$t=10$}
	\label{fig:test3sec10}
\end{figure}

\begin{figure}[]
	\includegraphics[width=0.45\linewidth]{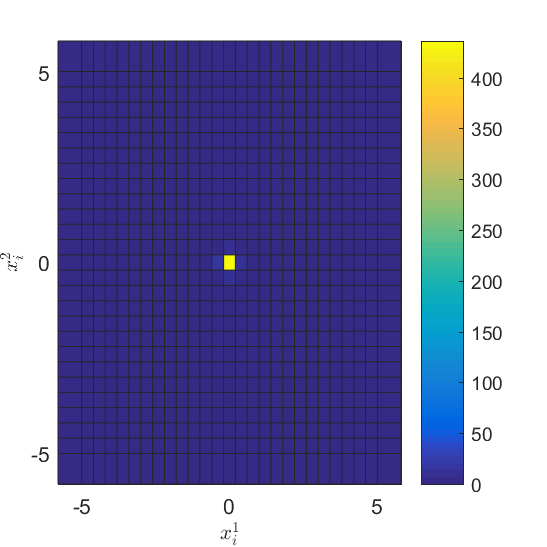}
	\includegraphics[width=0.45\linewidth]{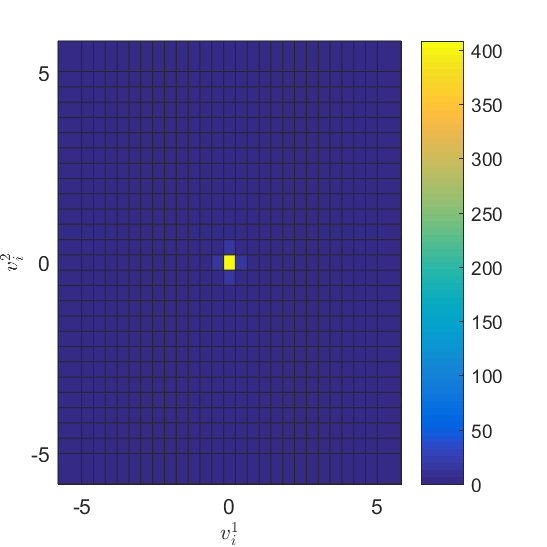}
	\caption{$t=15$}
	\label{fig:test3sec15}
	\includegraphics[width=0.45\linewidth]{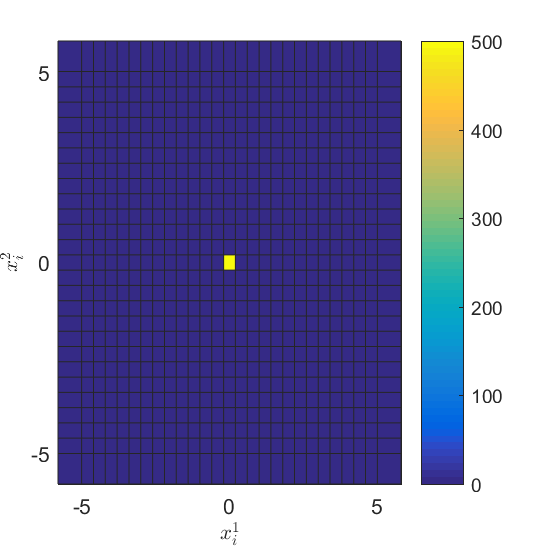}
	\includegraphics[width=0.45\linewidth]{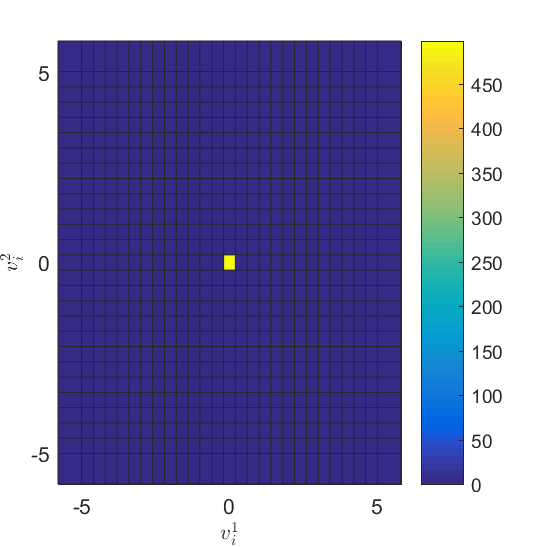}
	\caption{$t=20$}
	\label{fig:test3sec20}	
\end{figure}

%
%
%
%
%
%
%
\section{Conclusion and future work}
In this paper, we reinterpreted the rate of return and the favorability  as phase point on phase space, and derive an agent based particle model for herding phenomena. We then use this model to
prove the herding behavior induced by various dynamic market signals.  We also provided various numerical simulations to verify and visualize these results.
This work can be developed or extended in several direction.
First, particle model with noise is definitely on the first place in the next-to-do list. We restricted ourselves to noiseless situation for clarity in this paper.
Secondly, upon  incorporating collision avoidance mechanism, our model can be naturally modified to model herding or swarming behavior of particles, animals, bacteria, individuals, unmanned vehicles and so on.
Thirdly, we didn't consider the case where the players can enter or leave the market based on the asset prices. This also seems to be interesting possible future work.
Finally, kinetic and hydrodynamic limit of this model and verifying herding phenomena at those levels will be be treated in future works.

\textbf{Acknowledgement}. The authors thank Jane Yoo for serious discussion and helpful comments.


\begin{thebibliography}{10}


\bibitem{ABHKL13} S. Ahn, H.-O. Bae, S.-Y. Ha, Y. Kim and H. Lim, Application of flocking mechanism to the modeling of stochastic volatility, {\it Math. Models Methods Appl. Sci.} {\bf 23} (2013) 1603--1628.

\bibitem{ACHL} S. Ahn, H. Choi, S.-Y. Ha, H. Lee,  On collision-avoiding initial configurations to Cucker-Smale type flocking models, {\it Commun. Math. Sci.} {\bf 10} (2012) 625--643.

\bibitem{BHKLLY15} H.-O. Bae, S.-Y. Ha, Y. Kim,
 S.-H. Lee, H. Lim, and  J. Yoo,
A mathematical model for volatility flocking with a regime switching mechanism in a stock market,
{\it Math. Models Methods Appl. Sci.} {\bf 25} (2015) 1299--1335.

\bibitem{Abreu_B03} D. Abreu and M. K. Brunnermeier,
Bubbles and Crashes,
{\it Econometrica} {\bf 71} (2003) 173--204.

\bibitem{AZ98} C. Avery and P. Zemsky, Multidimensional uncertainty and herd behavior in financial markets, {\it Amer. Eco. rev.} (1998) 724--748.


\bibitem{Banerjee92} A. V. Banerjee, A simple model of herd behavior,
{\it Quar. J. Eco.} {\bf 107} (1992) 797--818.

\bibitem{BTTYB09}
A. B. T. Barbaro, K. Taylor, P. F. Trethewey, L. Youseff and B. Birnir, Discrete and
continuous models of the dynamics of pelagic fish, {\it Math. Comput. Simulat.} {\bf 79} (2009)
3397--3414.

\bibitem{BBNS10}
 N. Bellomo, A. Bellouquid, J. Nieto and J. Soler, Multiscale biological tissue models
and flux-limited chemotaxis for multicellular growing systems, {\it Math. Models Methods
Appl. Sci.} {\bf 20} (2010) 1179–-1207.

\bibitem{BHW92} S. Bikhchandani, D. Hirshleifer and I. Welch,
A theory of fads, fashion, custom and cultural change as informational cascades,
{\it J. Poli. Eco.} {\bf 100} (1992) 992--1027.

\bibitem{Brun01} M. K. Brunnermeier,
{\it Asset Pricing under Asymmetric Information,
Bubbles, Crashes, Technical Analysis and Herding}
(Oxford University Press on Demand, 2001).

\bibitem{BMP} M. Burger, P. Markowich, and J.-F. Pietschmann. Continuous limit of a crowd motion and herding
model: analysis and numerical simulations, {\it Kinet. Relat. Models} {\bf4} (2011) 1025--1047.

\bibitem{CCMP} J. A. Carrillo, Y.-P. Choi, P. B. Mucha, J. Peszek, Sharp conditions to avoid collisions in singular Cucker-Smale interactions, {\it Nonlinear Anal. Real World Appl.} {\bf37} (2017) 317–-328.

\bibitem{CFRT} J. A. Carrillo, M. Fornasier, J. Rosado and G. Toscani, Asymptotic flocking dynamics for the kinetic Cucker-Smale model, {\it SIAM J. Math. Anal.}  {\bf 42} (2010) 218--236.

\bibitem{CCGPSST10}
A. Cavagna, A. Cimarelli, I. Giardina, G. Parisi, R. Santagati, F. Stefanini and
R. Tavarone, From empirical data to inter-individual interactions: Unveiling the rules
of collective animal behavior, {\it Math. Models Methods Appl. Sci.} {\bf 20} (2010) 1491-–1510.

\bibitem{CHL}
Y.-P. Choi, S.-Y. Ha and Z. Li, Emergent dynamics of the Cucker-Smale flocking model and its variants.  {\it Active particles. Vol. 1. Advances in theory, models, and applications, 299--331,
Model. Simul. Sci. Eng. Technol.} (Birkhäuser/Springer, Cham, 2017).


\bibitem{CS} F. Cucker and S. Smale, Emergent behavior in flocks, {\it IEEE Trans. Automat. Control}  {\bf 52}  (2007) 852–-862.

\bibitem{CS2} F. Cucker and S. Smale, On the mathematics of emergence, {\it Japan. J. Math.} {\bf 2} (2007) 197-–227.


\bibitem{DL} M. Delitala and T. Lorenzi, A mathematical model for value estimation with public information and
herding, {\it Kinet. Relat. Models} {\bf7} (2014) 29--44.

\bibitem{DW}  A. Devenow and I. Welch, Rational herding in financial economics, {\it Europ. Econ. Rev.} {\bf 40} (1996) 603--615.

\bibitem{DJT} B. During, A. Jungel and L. Trussardi, A kinetic equation for economic value estimate with irrationality and herding, {\it Kinet. Relat. Models}  {\bf10} (2017) 239--261.


\bibitem{DM} P. Degond and S. Motsch, Continuum limit of self-driven particles with orientation interaction, {\it Math. Models Methods Appl. Sci.} {\bf 18} (2008) 1193–-1215.



\bibitem{HLL} S.-Y. Ha, K. Lee and D. Levy, Emergence of time-asymptotic flocking in a stochastic Cucker-Smale system, {\it Commun. Math. Sci.} {\bf 7} (2009) 453–-469.

\bibitem{HL} S.-Y. Ha, J. G. Liu, A simple proof of the Cucker-Smale flocking dynamics and mean-field limit, {\it Commun. Math. Sci.} {\bf 7} (2009) 297–-325.

\bibitem{HT} S.-Y. Ha and E. Tadmor : From particle to kinetic and hydrodynamic descriptions of flocking, {\it Kinet. Relat. Models} {\bf 1} (2008) 415–-435.

\bibitem{HemSuk09} C. S. Hemphill, J. Suk, The law, culture, and economics of fashion, {\it Stan. L. Rev.} {\bf 61} (2009) 1147--1200.



\bibitem{HwangSalmon04} S. Hwang and M. Salmon, Market stress and herding, {\it J. Empir. Finan.} {\bf 11} (2004) 585--616.

\bibitem{Lee95} I. H. Lee, Market crashes and informational avalance, {\it The Rev. Econ. Stud.} {\bf 65} (1995).

\bibitem{Merton} R. C. Merton and P. A. Samuelson, {\it Continuous-Time Finance} (Blackwell
Publishing, 1992).

\bibitem{Milgrom_S82} P. R. Milgrom and N. Stokey, Information, Trade and Common Knowledge,
{\it J. Econ. Theory} {\bf 26} (1982) 17--27.

\bibitem{MT} S. Motsch and E. Tadmor, A new model for self-organized dynamics and its flocking behavior, {\it J. Stat. Phys.} {\bf 144} (2011) 923-–947.

\bibitem{Santos_W97} M. S. Santos and M. Woodford, Rational asset pricing bubbles,
{\it Econometrica} {\bf 65} (1997) 19--57.

\bibitem{Tirole82} J. Tirole, On the possibility of speculation under rational expectations,
{\it Econometrica} {\bf 50} (1982) 1163--1182.

\bibitem{Toscani}  G. Toscani, Kinetic models of opinion formation. Commun. Math. Sci. {\bf4} (2006), 481--496.
\bibitem{VCBCS} T. Vicsek, ; A. Czirók, ; E. Ben-Jacob, I. Cohen and O. Shochet, Novel type of phase transition in a system of self-driven particles, {\it Phys. Rev. Lett.}  {\bf 75}  (1995) 1226–-1229.

\bibitem{TT} J. Toner and Y. Tu,  Flocks, herds, and schools: a quantitative theory of flocking, {\it Phys. Rev.} {\bf E58}  (1998)  4828–-4858.


\bibitem{L} J. P. LaSalle, Some extensions of Liapunov's second method, {\it IRE Trans.} {\bf 7} (1960) 520--527.

\end{thebibliography}
\end{document}